\documentclass[12pt]{article}
\usepackage[english]{babel}
\usepackage{fullpage}
\usepackage{amssymb, amsmath}
\usepackage{graphicx}
\usepackage{bbm}
\usepackage{appendix}

\usepackage{enumerate}
\usepackage{url}
\usepackage{color}
\usepackage{amsthm, bm}
\usepackage{float}
\usepackage{booktabs, multirow,ltablex}
\usepackage{pdflscape}
\usepackage{authblk}

\usepackage{csquotes}
\usepackage[
backend=biber,
giveninits=true,
uniquename=false,
date=year,
style=authoryear-comp,
uniquelist=false,
sorting=nty,
url=false,
isbn=false,
doi=false,
eprint=false,
maxbibnames=10,
minbibnames=1,
maxcitenames=2,
mincitenames=1
]{biblatex}
\DeclareNameAlias{sortname}{family-given}
\newtheorem{theorem}{Theorem}

\newtheorem{assumption}{Assumption}

\newtheorem*{theorem*}{Theorem}
\newtheorem*{corollary*}{Corollary}
\newtheorem*{lemma*}{Lemma}
\newtheorem*{assumption*}{Assumption}

\DeclareMathOperator*{\argmin}{\arg\!\min}

\def\diag{\mathop{\mathrm{diag}}}
\def\Tr{\mathop{\mathrm{tr}}}
\def\tr{\mathop{\mathrm{tr}}}

\def\SG{\mathop{\mathrm{SG}}}
\def\SGV{\mathop{\mathrm{SGV}}}

\def\E{\mathop{\mathbb{E}}}
\def\R{\mathop{\mathbb{R}}}

\renewcommand{\theenumi}{\arabic{enumi}}
\renewcommand{\theenumii}{\arabic{enumii}}

\makeatletter
\renewcommand{\p@enumii}{\theenumi.}
\renewcommand{\p@enumiii}{\theenumi.\theenumii.}
\makeatother

\usepackage{lastpage}

\title{Inference for Heterogeneous Graphical Models using Doubly High-Dimensional Linear-Mixed Models}

\author[1]{Kun Yue}
\author[1]{Eardi Lila}
\author[1,2]{Ali Shojaie}

\affil[1]{Department of Biostatistics, University of Washington, Seattle}
\affil[2]{Department of Statistics, University of Washington, Seattle}

\date{} % Removes date

\begin{document}

\maketitle

\begin{abstract}%  
\noindent Inferring the structure of population-level functional connectivity networks from multi-subject functional magnetic resonance imaging data is essential for understanding the connectivity patterns underlying cognitive and behavioral functions. Graphical models provide a formal statistical framework for inferring the graph structure. However, the multi-level nature of functional magnetic resonance imaging data---where each level corresponds, for instance, to a subject---compromises the inferential validity of most existing methods. In this paper, we develop a valid inference framework for high-dimensional graphical models that accounts for cross-level heterogeneity. We introduce a neighborhood-based method to learn the graph structure and reframe the problem as that of inferring fixed effect parameters in a doubly high-dimensional linear mixed model. We then propose a LASSO-based estimator and a de-biased LASSO-based inference framework for these fixed effect parameters. We evaluate the performance of the proposed framework using benchmark simulation studies and apply it to descriptively compare graph structures between the functional connectivity networks of two predefined groups of subjects from the Human Connectome Project.

\end{abstract}

%%%%%%%%%%%%%%%%%%%%%%%%%%%%%%%%%%%%%%%%%%%%%%%%%%%%%%%%%%%%%%%%%%%%%%%%%%%%%%

% \begin{keywords}
%   high-dimensional random effect, 
% heterogeneous network,
% neighborhood selection,
% functional connectivity network,
% de-biased LASSO inference
% \end{keywords}

\newpage

%%%%%%%%%%%%%%%%%%%%%%%%%%%%%%%%
\section{Introduction}
\label{section:intro}
%%%%%%%%%%%%%%%%%%%%%%%%%%%%%%%%
Multi-subject resting-state functional magnetic resonance imaging (fMRI) data are increasingly used to gain insight into differences in brain functional connectivity networks associated with lifestyle, mental health, and neurological conditions \parencite{smith2011network}. A prime example is the Human Connectome Project \parencite[HCP;][]{van2013wu}, whose fMRI data have been used to detect subgroup differences in functional connectivity linked to substance use, sleep quality, body mass index, and cognitive ability, among others \parencites{curtis2016sleep, doucet2018elevated,cheng2019decreased, soleimani2023altered}. Many such studies hinge on the ability to characterize functional neural connections within predefined subgroups, requiring methods that enable reliable inference of population-level functional connectivity. However, fMRI data from HCP and similar projects are intrinsically multi-level, consisting of observations over time for each individual, as well as observations across multiple individuals, violating the assumptions of many standard inference methods for graphical models. 

Figure~\ref{fig:example} illustrates our application setting: For each subject (i.e., level) $i= 1,\ldots,n$, resting-state whole-brain fMRI signals provide indirect measures of the neuronal activation levels at multiple brain locations over time. Standard pre-processing leads to spatially distributed maps that define a set of brain nodes $V = \{1,\ldots,p\}$, with details to be described in Section~\ref{section:data}. Each brain region has an associated fMRI signal describing its activation pattern over time, denoted by $Y^i_j$. In this work, we introduce an inference framework for testing whether pairs of brain regions exhibit non-zero correlation at the population level, while controlling for the effects of all other brain regions. This enables the study of subgroup differences in the direct connectivity structure of functional brain networks.

Gaussian graphical models (GGMs) are powerful tools for modeling conditional dependence relationships among the set of brain regions (i.e., nodes) $V$, enabling inference of the graph structure underlying observed functional activation patterns. GGMs have been extensively studied in high-dimensional settings, with applications in neuroscience \parencite{ng2013novel, monti2017learning} and genomics \parencite{krumsiek2011gaussian, zhao2019cancer}, among other fields. Two of the most popular approaches for independent observations are the graphical lasso \parencite{yuan2007model, friedman2008sparse} and neighborhood selection \parencite{meinshausen2006high}. Other graph structure learning methods include 
greedy search \parencite{ray2015improved, bresler2015efficiently}, structured regularization \parencite{cai2011constrained, defazio2012convex} and regularized score matching \parencite{lin2016}. Recent developments in high-dimensional graphical modeling have also considered non-Gaussian observations \parencite{liu2012high, voorman2014graph, yu2019generalized} and functional data \parencite{qiao2019functional, solea2020copula, waghmare2023functional}. Inference for graphical models has also received increasing attention in recent years, enabling tests of whether non-zero dependencies between brain regions reflect real connections or if differences in connectivity between patient groups indicate true population-level differences \parencite{shojaie2020differential}. Examples include multiple testing with asymptotic control of false discovery rates \parencite{liu2013gaussian} and direct testing of edge weights based on the asymptotic normality of different (de-biased) $\ell_1$-regularized estimators \parencite{jankova2017honest, ren2015asymptotic}. See \textcite{jankova2018inference} for a detailed review. 

\begin{figure}[!htb]
    \centering
    \includegraphics[width=0.9\textwidth]{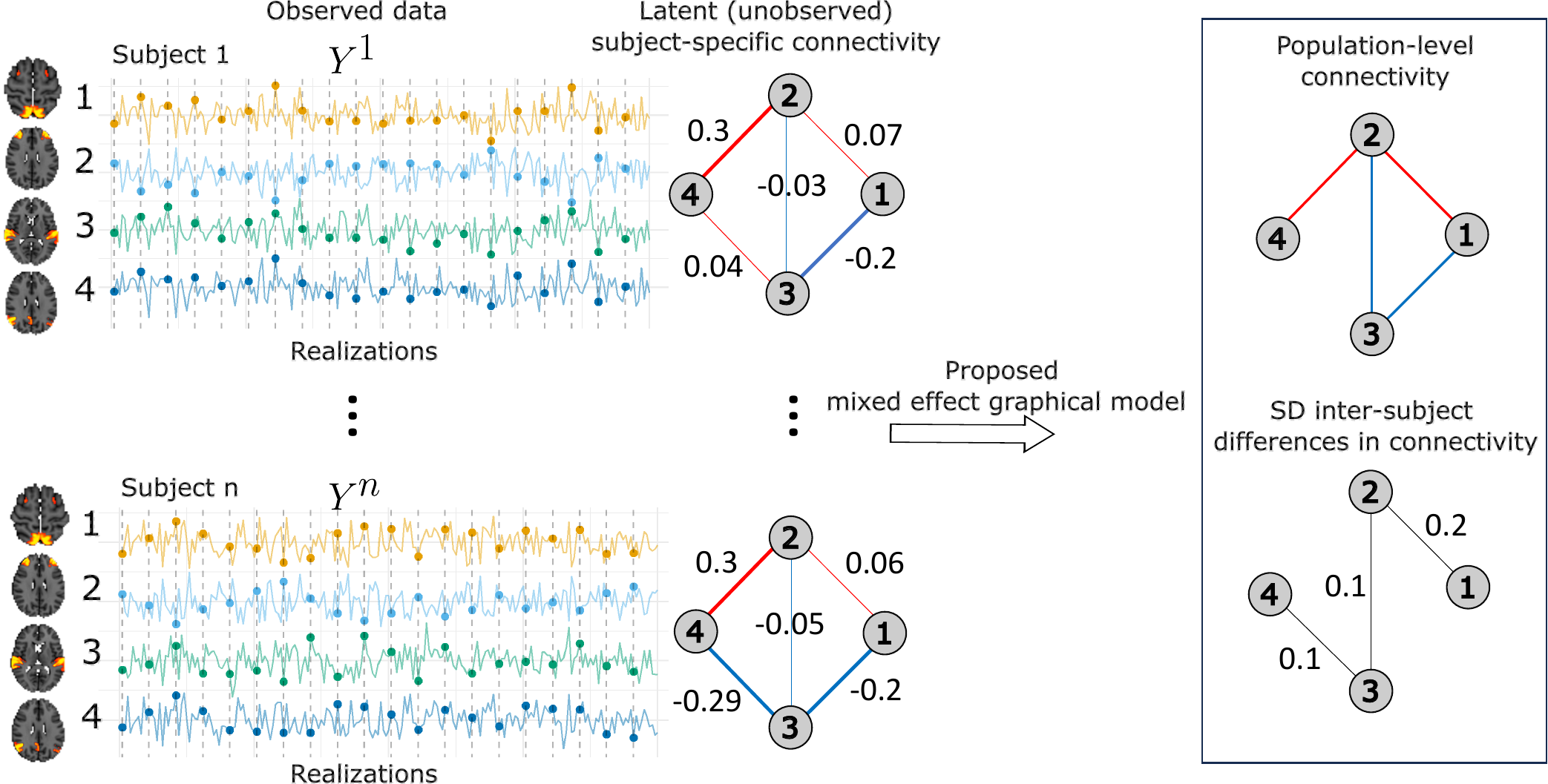}
    \caption{\footnotesize Illustration of the proposed functional connectivity brain network analysis. \textit{Left panel}: fMRI signals measuring the activation level at each brain node for individual subjects. \textit{Middle panel}: Latent subject-level brain networks underlying the observed fMRI data. The lines between the brain nodes depict connectivity, with the numbers denoting the connectivity strength, red lines representing positive connectivity, and blue lines representing negative connectivity. \textit{Right panel}: The outputs of our proposed graphical model. These are the estimates of the underlying \textit{population-level} brain network and the associated network of standard deviations (SD) summarizing inter-subject variation around the population-level brain network.}
    \label{fig:example}
\end{figure}

Learning the functional connectivity graph structure from the fMRI data ${Y^i_j}$ in our application presents two primary challenges that existing methods fail to address: (i) the data have a \emph{clustered or multi-level structure}, where each cluster, or level, corresponds to the observations of a specific subject; and (ii) fMRI observations over time for each single brain node typically exhibit temporal correlation. While the latter can be mitigated through various whitening procedures, including model-based pre-whitening procedures \parencite{olszowy2019accurate, woolrich2001temporal, yue2024iterative} or simple downsampling approaches, the complications due to the clustered structure of the data have not been extensively studied in this setting. Many neuroscience studies ignore the \emph{heterogeneity} inherent in multi-level data and simply infer a single graphical model for all subjects \parencite{dyrba2020gaussian}. This assumes a fixed dependence structure for all the subjects, which is contrary to a growing body of evidence that points to considerable subject-level heterogeneity in functional connectivity networks \parencite{mumford2006modeling, monti2017learning}. Such heterogeneity cannot be easily addressed with resampling techniques \parencite[e.g.,][]{narayan2016mixed} which lack theoretical guarantees for type-I error control and are computationally demanding when the number of brain regions is large. Another popular approach is to employ a two-stage strategy: in the first stage, separate graphical models are inferred for each subject; in the second stage, individual-level summary statistics are used for group-level analysis \parencite{narayan2016mixed,  deshpande2009multivariate, morgan2011cross, qiu2022inference}. P-value aggregation via Fisher's method \parencite{deshpande2009multivariate} and a t-test based on individual-level statistics \parencite{morgan2011cross} are two typical examples. While straightforward, such methods ignore any shared brain network information across subjects, which can lead to inefficient estimation and inference. Moreover, they can lead to conflicting conclusions from different second-stage aggregation choices, as well as erroneous conclusions due to not properly accounting for the uncertainty in first-stage estimates \parencite{chiang2017bayesian}.

We illustrate the limitations of these two-stage approaches through a simple toy example depicted in Figure~\ref{fig:toy}. In this toy example, we infer the connectivity between a given node and six other nodes using the neighborhood selection approach from a marginal model perspective (see Figure~\ref{fig:toy} for details). Specifically, the fixed GGM method that ignores the heterogeneity can result in false discoveries, even if only the population-level average network is of interest. Two-stage approaches may lead to conflicting conclusions, and may result in both inflated type-I errors and/or reduced power.

\begin{figure}[!t]
\centering
\begin{minipage}{0.25\textwidth}
  \centering
  \includegraphics[width=\linewidth]{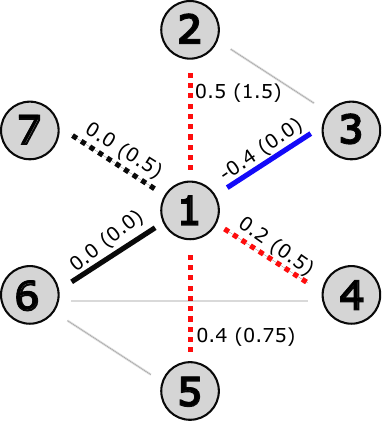}
  % \subcaption{(optional) Visualization}
  % \label{fig:toy:a}
\end{minipage}\hfill
\begin{minipage}{0.74\textwidth}
  \centering
  {\footnotesize
  \begin{tabular}{c c c c c c c}
    & \multicolumn{6}{c}{Edge}\\
    \cmidrule(lr){2-7}
      & 1---2   &   1---3&  1---4&  1---5 &  1---6 &  1---7\\
      \midrule
      Fixed effect coeff $\beta$ & 0.50 & -0.40 &  0.20 &  0.40 &  0.00 &  0.00 \\
      SD of random effect coeff  & 1.50 & 0.00 & 0.50 & 0.75 & 0.00 & 0.50 \\
      \midrule \\
         Approach &  \multicolumn{4}{c}{Power} & \multicolumn{2}{c}{Type-I error}\\
      \cmidrule(lr){1-1} \cmidrule(lr){2-5} \cmidrule(lr){6-7}
       Two-stage t-test & 0.255 & 0.795 & 0.165 & 0.460 & 0.050 & 0.040 \\
      Two-stage Fisher's method & 1.000 & 0.380 & 0.660 & 0.925 & 0.075 & \underline{0.500} \\
    Fixed GGM & 0.750 & 0.755 & 0.310 & 0.605 & 0.035 & \underline{0.150} \\
    \textbf{Mixed effect GGM} & \textbf{0.285} & \textbf{0.980} & \textbf{0.240} & \textbf{0.505} & 0.025 & 0.055 \\
      \bottomrule
      \end{tabular}
      }

    \end{minipage}
    \caption{\footnotesize Toy example illustrating the inference of group-level conditional dependence using a neighborhood selection approach, targeting the connections between node 1 and the remaining six nodes. Blue lines denote negative connectivity, red lines denote positive connectivity, and black lines denote zero connectivity. Dashed lines denote connectivity edges with inter-individual variation, whereas solid lines indicate those without. The numbers represent average connectivity strengths across subjects, with the associated standard deviations (SDs) shown in parentheses. We infer connectivity using two-stage approaches, the fixed GGM approach and the proposed mixed effect GGM approach. We generate data from a marginal model following model \eqref{model.def.2}, where the fixed effect coefficients and the standard deviation of the corresponding random effect coefficients are shown in the top table. The noise terms are independently and identically distributed as standard normal random variables. For each of the 20 subjects, we simulate 10 observations. We then use the neighborhood selection approach to infer the \textit{presence} of network edges. For the fixed GGM approach, we concatenate observations from all subjects and fit a single GGM using a simple linear regression model. For two-stage approaches, we use simple linear regression in the first stage for each subject, then in the second stage use either Fisher's method to aggregate p-values \parencite{deshpande2009multivariate} or use a t-test on each subject's coefficient estimates \parencite{morgan2011cross}. For the proposed mixed effect GGM approach, we use a linear mixed-effect model with each subject treated as a cluster (see Section~\ref{section:method} for model details). Power/Type-I errors are computed based on 500 replications. The simulations show that when there is high subject-level heterogeneity, two-stage approaches may have highly-inflated type-I error (Fisher's method for edge 1---7), or have lower power to detect dependence (t-test method) compared to the mixed effect GGM approach. The fixed GGM approach shows an inflated type-I error for edge 1---7. \textit{The results show that the proposed mixed-effect GGM provides the highest power among methods that control type-I error}.
    }
\label{fig:toy}
\end{figure}

To address the limitations of existing methods for inferring population-level brain connectivity networks while accounting for subject-level heterogeneity, in Section~\ref{section:method}, we propose a \textit{mixed effect GGM}. Building on a neighborhood-based estimation framework similar to that of \textcite{narayan2016mixed}, our approach models each edge as a subject-level random effect centered around a population-level mean. The key novelty is the estimation and inference approach: we recast the proposed model as a \emph{doubly high-dimensional linear mixed model}, where fixed effects represent population-level edges and random effects capture subject-specific deviations. In this setting, both the number of fixed and random effects parameters can be larger than the sample size, and the fixed and random design matrices in the resulting linear mixed model are identical. These factors violate the assumptions of existing methods, which we address using penalized estimation and a de-biased inference procedure that circumvents the need to estimate variance components when performing inference on fixed effects. Simulations demonstrate that although our proposed approach is primarily motivated by clustered or multi-level structures in fMRI data, it is also robust to the presence of temporal correlation.

In Section~\ref{section:data}, we apply the proposed method to infer and descriptively compare functional connectivity networks from HCP fMRI of individuals with recent cannabis use (i.e., a positive tetrahydrocannabinol (THC) screen) and a control group that tested negative for THC. We evaluate the performance of the proposed approach via extensive simulations in Section~\ref{section:simulation}. Finally, we provide an estimation strategy for the mixed effect variance components in Appendix~\ref{section:extension.vc} and present the theoretical results in Appendix~\ref{section:theory}. Proofs and additional simulation results are included in the Supplementary Materials.

\section{Method}
\label{section:method}

\subsection{Notation and problem setup}
We begin by introducing some notation. We denote an $m \times m$ identity matrix by $I_m$. For a matrix $A$, we denote by $A_{j,k}$ the $(j,k)$ entry of $A$, by $A_j$ the $j$th column of $A$, and by $A_{-j}$ the sub-matrix of $A$ obtained by dropping the $j$th column. Similarly, we use index sets $S$, $\{j,k\}$ and $1:j$ to denote multiple columns/entries in a matrix/vector, and use $-S$, $-\{j,k\}$ and $-\{1:j\}$ to denote a sub-matrix/sub-vector obtained by removing the indicated columns/entries. For a set of values $\{u_l\}_{l=1}^K$ and a set of matrices $\{A_l\}_{l=1}^K$, we let $\diag \left(\{u_l\}_{l=1}^K\right)$ be the diagonal matrix with $(l,l)$ entry $u_l$, and let $\diag\left(\{A_l\}_{l=1}^K\right)$ be the block-diagonal matrix with the $l$th block $A_l$.

Suppose that for each subject $i=1, \dots, n$ the observed data matrix of fMRI data, $Y^i\in \R^{m \times p}$, includes $m$ repeated observations for each of the $p$ nodes. Without loss of generality, we assume the observations for each node are centered at zero. For simplicity, we also assume an equal number of observations per subject; however, our results continue to hold if the $i$th subject has $m_i$ observations, provided that $c_1\max_i(m_i) \leq \min_i(m_i) \leq c_2 \max_i(m_i)$ for some fixed constants $c_1, c_2>0$. The population-level connectivity network is characterized by the fixed inverse covariance matrix $\Omega \in \R^{p \times p}$, whereas the subject-specific conditional independence networks are characterized by $\Omega^i\in\R^{p\times p}$. 

We aim to infer the zero pattern of the edges in the population-level network, while accounting for subject-level heterogeneity. To this end, we propose a neighborhood-based method where we model the subject-specific edge weights $\Omega^i_{j,k}$ as \emph{random variables}. Specifically, we assume the following neighborhood-based model for the fMRI data $Y^i$; for a brain node $j \in V$, 
\begin{align}
    \label{model.def}
    Y^i_j = \sum_{k=1, k\neq j}^p Y^i_{k}  b^i_{j,k} + \epsilon^i_{j}, \quad i=1, \dots,n,
\end{align}
where given $Y^i_{k}, k \ne j$, the subject-level connectivity coefficients, $b^i_{j,k}$, have mean $b_{j,k}$ and variance $\sigma^2_{j,k}$. The coefficients $b^i_{j,k}$ are proportional to the edge weights $\Omega^i_{j,k}$, i.e., $b_{j,k}^i=-\Omega_{j,k}^i/\Omega_{j,j}^i$. The randomness in $b^i_{j,k}$ captures the subject-level heterogeneity, while its mean $b_{j,k}$ captures the shared dependence structure across subjects. Our main goal is to test whether, conditioned on other nodes, a pair of nodes $(j,k)$ is functionally connected at the population level. Specifically, we test the null hypothesis $H_0: b_{j,k}=0$, which is equivalent to $H_0: \E[\Omega^i_{j,k}] = 0$, for instance, under the assumption that $\E[\Omega^i_{j,k}\mid \Omega^i_{j,j}] = c_{j,k}\,\Omega^i_{j,j}$,
for some constant $c_{j,k}$. This condition is satisfied, for example, when the subject-specific precision matrices follow a Wishart distribution.

The model in equation \eqref{model.def} can be seen as a linear mixed model (LMM), formulated as:
\begin{align}
\label{model.def.2}
    Y^i_j = Y^i_{-j} \beta_j + Y^i_{-j} \gamma^i_j + \epsilon^i_j, \quad i=1, \dots, n.
\end{align}
Here, we treat node $j$ as the outcome variable and the associated vectors $Y^i_j \in \R^{m}$ as the outcome vector of observations. We treat the rest of the $p-1$ nodes as covariates. The fixed effect coefficients $\beta_j \in \R^{p-1}$ correspond to $\left\{b_{j,k}\right\}_{k=1, k\neq j}^p$, which represent the (scaled) edge weights $\Omega_{j,-j}$. The random effect coefficients $\gamma^i_j = \left\{b^i_{j,k}-b_{j,k}\right\}_{k=1, k\neq j}^p$ represent the subject-level variation of these $p-1$ (scaled) weights. The matrix $Y^i_{-j} \in \R^{m \times (p-1)}$ serves as both the fixed and random effect design matrices.

In the context of the proposed application to inferring connectivity, the fixed effect $\beta_{j,k}$ captures the population-level linear dependence between node $j$ and node $k$, conditional on all other $p - 2$ nodes. In other words, it represents direct connectivity. The associated random coefficient $\gamma^i_{j,k}$ captures deviations in the conditional linear dependence between nodes $j$ and $k$ at the level of subject $i$, reflecting subject-specific differences in brain connectivity. Treating subject-specific differences in connectivity as random effects allows us to model individual variability without estimating a separate fixed effect for each subject. This approach reduces the number of parameters that need to be estimated and allows us to borrow strength across subjects.

The hypothesis $H_0: \ b_{j,k} =0$ is thus equivalent to $H_0^\prime: \ \beta_{j,k}=0$, allowing us to recast the problem of determining whether a pair of brain nodes $(j,k)$ exhibits nonzero conditional linear dependence as that of testing $H_0^\prime: \beta_{j,k} = 0$, i.e., estimating and inferring the fixed effect coefficients in an LMM. Since the number of brain nodes $p$ is typically large in brain connectivity studies, the resulting LMM is \emph{doubly high-dimensional}, i.e., both fixed and random effects are high-dimensional. 

Framing our fMRI problem as an LMM clarifies its relation to the broader LMM literature, but this reformulation alone does not solve it. Indeed, despite its increasing relevance in applications, rigorous estimation and inference procedures for doubly high-dimensional LMMs are lacking. Methods for LMMs with high-dimensional fixed effects and low-dimensional random effects \parencite{lin2020statistical, bradic2020fixed} are not readily extendable to this setting. 
In addition, while Expectation-Maximization has been proposed for estimation \parencite{monti2017learning}, the consistency of the resulting estimator has not been thoroughly examined. Likewise, while consistent, valid inference for the estimator of \textcite{li2018doubly} has not been explored. Most related to our setting is the recent work by \textcite{li2021inference}, denoted \emph{LCL} hereafter. The \emph{LCL} framework includes a de-biased LASSO-based inference framework for doubly high-dimensional linear mixed models; however, as common in the LMM literature, it assumes a form of independence between the fixed and random effect design matrices---the fixed effect design matrix is assumed to have zero mean conditional on the random effect design matrix. This restrictive assumption constitutes a key shortcoming in our graphical modeling application, where the fixed and random effect design matrices are identical, leading to a violation of the zero conditional mean assumption in \emph{LCL}. This assumption can also be restrictive in many other applications, whenever the fixed effects and the random effects share a non-empty set of covariates \parencite{li2018doubly}. Moreover, the \emph{LCL} framework requires the number of random effects to be no larger than the number of observations per subject, which further limits its applicability.

Motivated by the neighborhood-based modeling of multi-level fMRI data in equation \eqref{model.def}, in this paper, we develop a new estimation and inference framework for doubly high-dimensional LMMs in equation \eqref{model.def.2}. To this end, we propose a penalized estimation and inference framework for the fixed effect coefficients, which, in the context of graphical modeling of fMRI data, allows us to determine whether a nonzero conditional linear dependency exists between any pair of brain nodes at the population level, while simultaneously accommodating subject-specific variations through the random coefficients $\gamma_j^i$.

Our main methodological contribution is to recast connectivity inference as inference on the fixed effect coefficients of an LMM, using an estimator closely related to \textcite{li2021inference}. This formulation, however, requires two key extensions to enable valid inference in heterogeneous GGMs: (i) Our approach accommodates a larger number of random effects than the number of observations per subject, which, in the context of fMRI data, means that the number of brain regions can exceed the number of fMRI time measurements. This feature is especially useful when downsampling is employed to reduce the strength of temporal correlation. (ii) Our method does not assume conditional independence between the fixed and random effects design matrices, but instead imposes only minimal assumptions. This allows us to accommodate the mixed-effects model in equation~(\ref{model.def.2}), which naturally arises from the proposed heterogeneous neighborhood-based model of fMRI data and, in fact, features identical fixed and random effects design matrices. 

The above extensions provide the first mixed-model inference framework for learning functional connectivity networks in multi-level settings. 
Our empirical results show that the proposed method is also robust to violations of its theoretical assumptions. In particular, although our theoretical analysis assumes a more idealized scenario in which data are generated from the model in equation~(\ref{model.def.2}), we demonstrate through simulations that the proposed approach performs well when applied to data generated from a joint random graph model for $Y^i$, which cannot be trivially mapped to the conditionally specified model in equation~(\ref{model.def.2}). Specifically, when each $Y^i$ follows a matrix-normal distribution with temporal dependence among the $m$ observations, and its column covariance matrix $\Sigma^i$ is drawn from a Wishart distribution with mean $\Sigma$, where $\Omega = \Sigma^{-1}$ is sparse, our method can test whether individual entries of $\Omega$ are zero while controlling the type I error rate. This enables characterization of the population-level partial correlation structure. Moreover, even though the theoretical results assume no temporal correlation among observations for each subject, the proposed method continues to provide reasonable type-I error and power when temporal observations are correlated. 

% Together, these extensions provide the first mixed-model inference framework for learning functional connectivity networks in multi-level settings. 

\subsection{Mixed effect Gaussian graphical model}
\label{subsection:model_notation}

To achieve consistent estimation for the doubly high-dimensional LMM in equation \eqref{model.def.2}, we assume the true $\beta_j^*$ coefficients are sparse, with support $S_j$ and cardinality $s_j = |S_j|\ll p$. In the context of the fMRI data, this translates into population conditional linear dependencies between brain nodes being sparse, reflecting the biological hypothesis that the brain is organized into sparsely interconnected hubs \parencite{bullmore2012economy}. Let $\Sigma^i$ be the $p\times p$ subject-specific covariance matrix for subject $i$. 
%the rows of $Y^i$ are independent realizations of $N(0, \Sigma^i)$, for $i=1, \dots, n$. 
To allow for subject-level heterogeneity, we assume the subject-specific covariance matrices $\Sigma^i$ are random, and we denote by $\Sigma$ the corresponding population-level covariance matrix. Conditional on the matrix $Y^i_{-j}$, the random effect coefficients $\gamma^i_j$ and the noise vectors $\epsilon^i_j$ are independent with mean zero and covariance $\Psi_j \in \R^{(p-1)\times (p-1)}$ and $R^i_j \in \R^{m \times m}$, respectively. %We assume $\gamma^i_j \in \SG(c_1 \|\Psi_j\|_2)$ and $\epsilon^i_j \in \SG(c_2 \|R^i_j\|_2)$, which implies that $Y^i_j - Y^i_{-j}\beta_j | Y^i_{-j} \in \SG(c_1 \|\Psi_j\|_2 + c_2\|R^i_j\|_2)$. 

Similar to other specifications of graphical models \parencite[see, e.g.,][]{voorman2014graph, chen2015selection}, the model in equation \eqref{model.def.2} specifies the conditional distribution of $Y_j$ as a function of other variables, $Y_{-j}$. In this model, the population-level network edge $(j,k)$ is characterized by the $\beta_{j,k}$ coefficient, such that $(j,k) \in E$ if and only if $\beta_{j,k}\neq 0$, providing a convenient specification of the conditional dependence structure while accounting for subject-level variability of the edges. In our theoretical analysis, we focus on developing an inference framework for the model in equation~\eqref{model.def.2}, and leave the characterization of joint probability distributions that are consistent with conditionally specified models \parencite{wang2008conditionally} to future research. However, simulations in Section~\ref{section:simulation_graph} demonstrate that our model also performs well when the data are generated from a realistic joint probability distribution.

The unknown random effect covariance matrices $\Psi_j$ pose challenges to the estimation and inference of the fixed effect coefficients $\beta_j$. Following the quasi-likelihood approach in \textcite{fan2012variable}, we use proxy matrices in place of the unknown covariance matrix. Specifically, we use the proxy matrix $\Sigma_a^{i, (j)} = a Y_{-j}^i (Y_{-j}^i)^\top + I_{m}$, with a fixed positive constant $a$, to approximate the covariance matrix of $Y^i_j$, which is $\Sigma^{i,(j)}_{\theta} = Y_{-j}^i \Psi_j (Y^i_{-j})^\top + R^i_j$, for $i=1, \dots, n$. We then use a LASSO estimator for the fixed effect coefficients $\beta_j$ that leverages the proxy matrix $\Sigma_a^{i,(j)}$ to ``decorrelate'' the observations. In addition, we propose an inference framework based on the de-biased LASSO method. As we will discuss later, our estimation and inference procedures for the coefficients $\beta_j$ do not depend on the estimates of the variance components $\Psi_j$ and $R_j^i$ and are, hence, well-defined.

\subsubsection{LASSO-based neighborhood selection}
\label{Sec:method_beta_est}
We propose a LASSO estimator for the fixed effect coefficients $\beta_j$ based on the de-correlated observations. To this end, let $\Sigma_a^{(j)} = \diag\left(\{\Sigma_a^{i,(j)}\}_{i=1}^n\right)$, and $Y$ be the $nm \times p$ matrix obtained by vertically stacking $\{Y^i\}_{i=1}^n$. Our proposed estimator $\hat\beta_j$ is given by
\begin{align}
    \label{eq:fixed_effect_est_definition}
    \hat\beta_j = \argmin_{\beta_j \in \mathbb{R}^{p-1}} \frac{1}{2\Tr\left((\Sigma_a^{(j)})^{-1}\right)}\left\| (\Sigma_a^{(j)})^{-1/2} ( Y_{j} - Y_{-j} \beta_j) \right\|_2^2 + \lambda_{a,j} \|\beta_j\|_1,
\end{align}
with tuning parameter $\lambda_{a,j}>0$. In Appendix~\ref{section:theory}, we show that under mild assumptions $\hat\beta_j$ is a consistent estimator of $\beta_j$, for a suitable choice of $\lambda_{a,j}$ and for any choice of constant $a$.

Taking $\Sigma_a^{(j)} = I_{nm}$ in \eqref{eq:fixed_effect_est_definition} leads to the classical LASSO estimator. This approach to estimating the structure of partial correlation networks from fMRI data is widely adopted in the neuroimaging literature \parencite{pervaiz2020optimising, smith2011network}. However, while it is easy to show that this estimator is also consistent, due to the correlation across observations, this misspecified estimator has sub-optimal rate of convergence \parencite{li2021inference}.

\subsubsection{Inference based on de-biased LASSO}
We next propose an inference framework for the coefficient $\beta_{j,k}$, $k\in V\backslash \{j\}$ based on the asymptotic normality of the de-biased LASSO estimator. This will allow us to make valid inferences about the population quantities $\beta_{j,k}$ from the associated estimates $\hat{\beta}_{j,k}$, thereby enabling testing of the hypothesis that the population-level conditional linear dependence between two brain nodes is zero. As opposed to the low-dimensional setting, which arises when the number of regions is small compared to the number of observations, the task here is complicated by the regularization term, which biases the estimates of $\beta_{j,k}$, as well as by the dependence structure among observations. Our de-biasing procedure builds on the idea in \textcite{zhang2014confidence}, which uses regularized regression to estimate the bias term of the LASSO estimator. 

For $k\in V\backslash \{j\}$, our de-biased estimator $\hat\beta_{j,k}^{(\mathrm{db})}$ is defined as
\begin{align}
\label{eq:dblasso}
    \hat\beta_{j,k}^{(\mathrm{db})} = \hat\beta_{j,k} +\frac{ \sum_{i=1}^n \left(\hat w^i_{j,k}\right)^\top \left(\Sigma_b^{i,(j,k)}\right)^{-1/2} \left(Y_j^i-Y^i_{-j}\hat\beta_{j}\right)}{ \sum_{i=1}^n \left(\hat w^i_{j,k}\right)^\top \left(\Sigma_b^{i,(j,k)}\right)^{-1/2} Y^i_k},
\end{align}
where the proxy covariance matrices $\Sigma_b^{i,(j,k)}$ are defined as
\begin{align*}
 & \Sigma_b^{i,(j,k)} = a Y^i_{-\{j,k\}}(Y^i_{-\{j,k\}})^\top + I_{m}, 
\end{align*}
with the same constant $a$ used in $\Sigma_a^{(j)}$; letting $\Sigma_b^{(j,k)} = \diag\left(\{ \Sigma_b^{i,(j,k)} \}_{i=1}^n \right)$, and using the tuning parameter $\lambda_{j,k}>0$, the projection-related terms $\hat\kappa_{j,k}$, $\hat w_{j,k} = \left((\hat w^1_{j,k})^\top, \ldots, (\hat w^n_{j,k})^\top \right)^\top$ are defined as
\begin{align*}
    & \hat\kappa_{j,k} = \argmin_{\kappa_{j,k} \in \mathbb{R}^{p-2}}\frac{1}{2 \Tr\left(\left(\Sigma_b^{(j,k)}\right)^{-1}\right)}\left\| \left( \Sigma_b^{(j,k)} \right)^{-1/2} \left(Y_{k} - Y_{-\{j,k\}} \kappa_{j,k}\right)\right\|_2^2 + \lambda_{j,k} \|\kappa_{j,k}\|_1 \\
    & \hat w_{j,k} = \left(\Sigma_b^{(j,k)} \right)^{-1/2} \left(Y_{k} - Y_{-\{j,k\}} \hat\kappa_{j,k}\right).
\end{align*} 

Here, the vector $\hat w_{j,k}$ is approximately the orthogonal complement of the projection of the vector $\left(\Sigma_b^{(j,k)} \right)^{-1/2} Y_k$ onto the space spanned by the columns of $\left(\Sigma_b^{(j,k)} \right)^{-1/2} Y_{-\{j,k\}}$, where the projection vector $\hat\kappa_{j,k}$ is computed via LASSO regression. We use the proxy matrix $\Sigma_b^{(j,k)}$ to ``decorrelate'' observations $Y_k$ and $Y_{-\{j,k\}}$ in the LASSO regression. Note that we define the proxy matrix $\Sigma_b^{(j,k)}$ differently from \emph{LCL} \parencite{li2021inference}. Specifically, it is constructed by excluding both nodes $j$ and $k$ from its definition. This modification is crucial to successfully establishing the asymptotic normality of the de-biased estimator $\hat\beta_{j,k}^{(\mathrm{db})}$, especially in the setting where the fixed effect design matrix has overlapping columns with the random effect design matrix.

Denoting by $z_{\alpha}$ the $\alpha$th quantile of a standard normal distribution, we can construct a two-sided $(1-\alpha)\times 100$\% confidence interval for the coefficient $\beta_{j,k}$ as 
%
% \label{eq:CI_def}
$    \hat\beta_{j,k}^{(\mathrm{db})} \pm z_{1-\alpha/2}  \sqrt{\hat{V}_{j,k}}$,
where $\hat V_{j,k}$ is a sandwich-type estimator of the variance of $\hat\beta^{(\mathrm{db})}_{j,k}$, defined as:
\begin{align}
\label{df.hat.v}
    \hat V_{j,k} = \frac{\sum_{i=1}^n \left|(\hat w_{j,k}^i)^\top \left(\Sigma_b^{i,(j,k)}\right)^{ -1/2}\left(Y_j^i-Y_{-j}^i\hat\beta_j\right)\right|^2}{\left|\sum_{i=1}^n (\hat w_{j,k}^i)^\top \left(\Sigma_b^{i,(j,k)}\right)^{ -1/2} Y_{k}^i \right|^2}. 
\end{align}
In the context of our fMRI application, such confidence intervals will be used to test whether the linear dependence between a pair of brain nodes $(j,k)$, conditional on all other nodes, is zero.

\subsubsection{Estimating cross-subject variability}
An appealing property of the proposed estimation and inference framework for the fixed effects is that we do not need to estimate the variance components $\theta_j = (\Psi_j, \{R^i_j\}_{i=1}^n)$. This is convenient when only the fixed effects are of interest. However, variance components may also be of independent interest. In our application of functional connectivity network analysis, a non-zero random effect variance component indicates the presence of subject-level heterogeneity in the connectivity between two brain nodes. In Appendix~\ref{section:extension.vc}, we provide moment-based estimators for the covariance $\Psi_j \in \R^{(p-1)\times(p-1)}$ of the random-effects coefficients $\gamma_j^i$ in the node-$j$ regression, and the within-subject noise covariance $R^i_j \in \R^{m\times m}$ for the error vectors $\epsilon_j^i$. We assume the within-subject noise terms have covariance $R^i_j=\sigma_e^2 I_m$. Moreover, we assume $\Psi_j$ is diagonal, i.e., $\Psi_j=\diag(\psi_j)$. For a pair of nodes $(j,k)$, the diagonal entry $\Psi_{j,(k,k)}$ summarizes the subject-specific variability in the conditional linear dependence between node $j$ and $k$. Therefore, plotting $\Psi_{j,(k,k)}$ provides a convenient visualization in our application.

%%%%%%%%%%%%%%%%%%%%%%%%%%%%%%%%%%%%%%%%%%%%%%%%%%%%%
\section{Analysis of HCP data}
\label{section:data}

We employ the proposed method to learn brain connectivity networks using the Human Connectome Project (HCP) resting-state fMRI data \parencite{van2013wu}. HCP has acquired high-quality functional and structural imaging data from healthy adult subjects in an effort to enhance the understanding of neural connectivity. Here, we analyze resting-state fMRI data from 160 subjects, of whom 80 tested positive for THC prior to scanning (THC+) and the remaining 80 tested negative (THC-). Each fMRI scan provides a signal with a temporal resolution of 0.72 seconds and a spatial resolution of 2-mm isotropic \parencite{smith2015positive}. Standard pre-processing steps were applied to the fMRI signals \parencite{smith2013resting}, including spatial normalization \parencite{glasser2013minimal} and artifact removal \parencite{griffanti2014ica,smith2015positive}. % through tools from FreeSurfer20 \parencite{fischl1999cortical}, FSL19 \parencite{jenkinson2012fsl}, and HCP workbench21 \parencite{marcus2013human}. 
Then a Group Independent Component Analysis was applied to generate $200$ spatially-distributed brain nodes \parencite{smith2014group}. Using a dual-regression, $200$ associated time series of length $1,200$ were obtained, each representing the activation pattern of a brain node over time (see \textcite{smith2015positive}, for more details). Each time series was uniformly downsampled to a length of $120$. This step is not strictly required, as our simulations show that the proposed inference framework is robust to within-subject temporal dependence stronger than that observed in the real data (Supplementary Figure~S6), likely owing to the sandwich-type variance estimator in Equation~\eqref{df.hat.v}. However, downsampling reduces computational time.

\begin{figure}[!htb]
    \centering
    \includegraphics[width=0.8\textwidth]{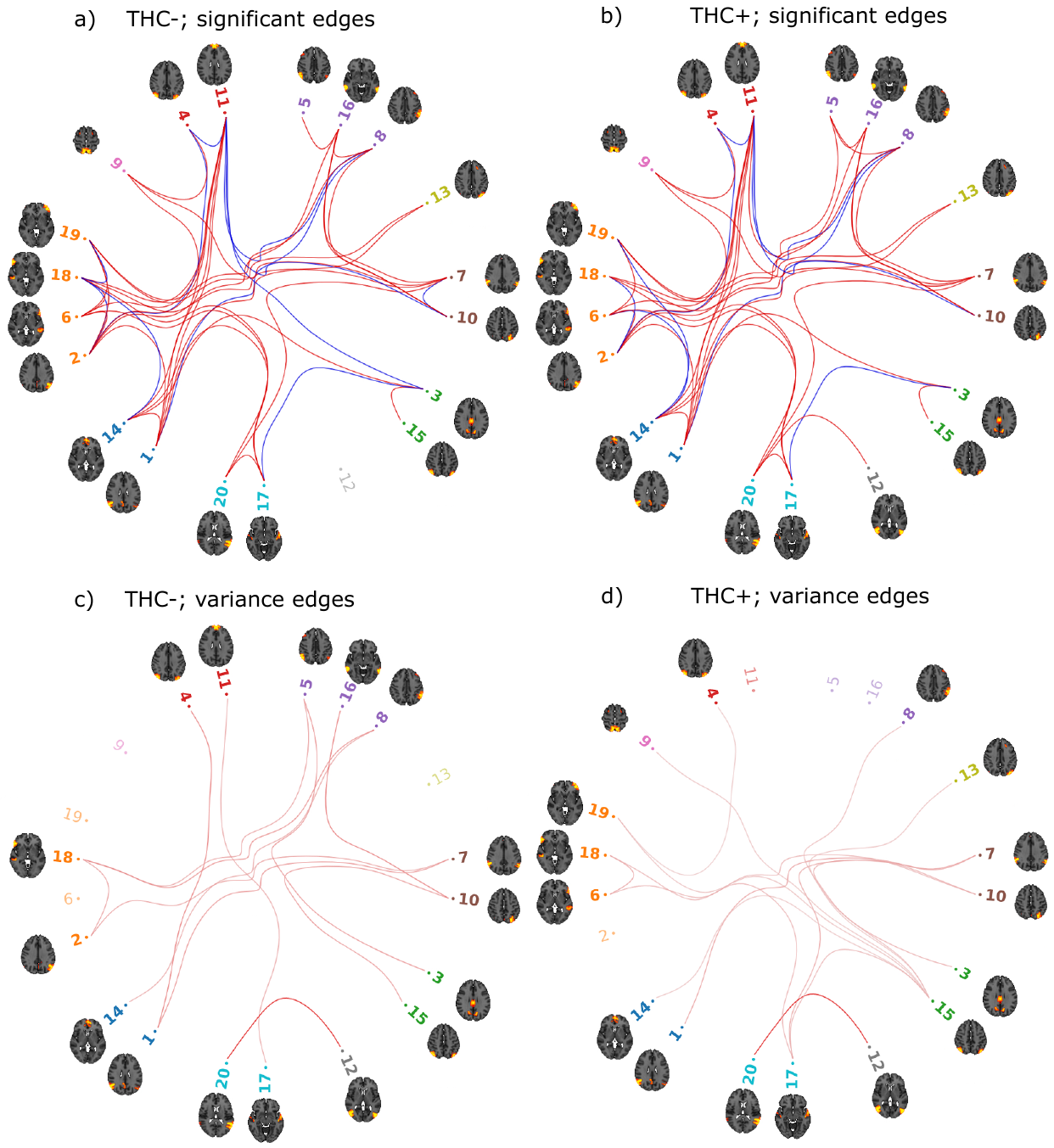}
    \caption{\footnotesize Estimated brain sub-networks and associated variances for the THC+ and THC- groups, as obtained using the proposed method. We present the results only for a subset of the 200 nodes that are associated with the default mode network, i.e., the brain regions that are active during passive rest. These were selected based on visual examination of the brain nodes. The node clusters in the networks are defined through hierarchical clustering based on the fixed effect estimates from the proposed method. 
    \textbf{a,b}: The most significant edges (adjusted p-values $< 1\times 10^{-13}$) in each group. Red edges represent positive edge strength, and blue edges represent negative edge strength.
    \textbf{c,d}: The edges in the top 25\% based on their variances, selected from those depicted in sub-plots \textbf{a} and \textbf{b}. The edge width is proportional to the estimate of edge variance, with thicker edges representing higher subject-level heterogeneity.}
    \label{fig:data_beta}
\end{figure}

We apply the proposed method to the THC+ and the THC- groups separately, estimating and testing the significance of the functional connectivity (network edges) between every pair of brain nodes in each group, and estimating the variance component for each network edge. For the latter, we assume a sparse diagonal covariance matrix for the random effects, and we split the samples into three approximately equal-sized sub-samples for estimation. To account for the asymmetry of the neighborhood selection approach, in each group we set the edge $E_{j,k}$ based on $\hat \beta^{(\mathrm{db})}_{j \leftrightarrow k} = (\hat\beta^{(\mathrm{db})}_{j,k} + \hat\beta^{(\mathrm{db})}_{k,j})/2$. We use
$\hat V_{j\leftrightarrow k}=(\hat V_{j,k}+\hat V_{k,j})/2$ as an upper bound for its variance. The p-value for testing the null hypothesis $H_0: \beta_{j \leftrightarrow k}=0$ is then computed based on these symmetrized quantities.  

\begin{figure}[!htb]
    \centering
    \includegraphics[width=0.8\textwidth]{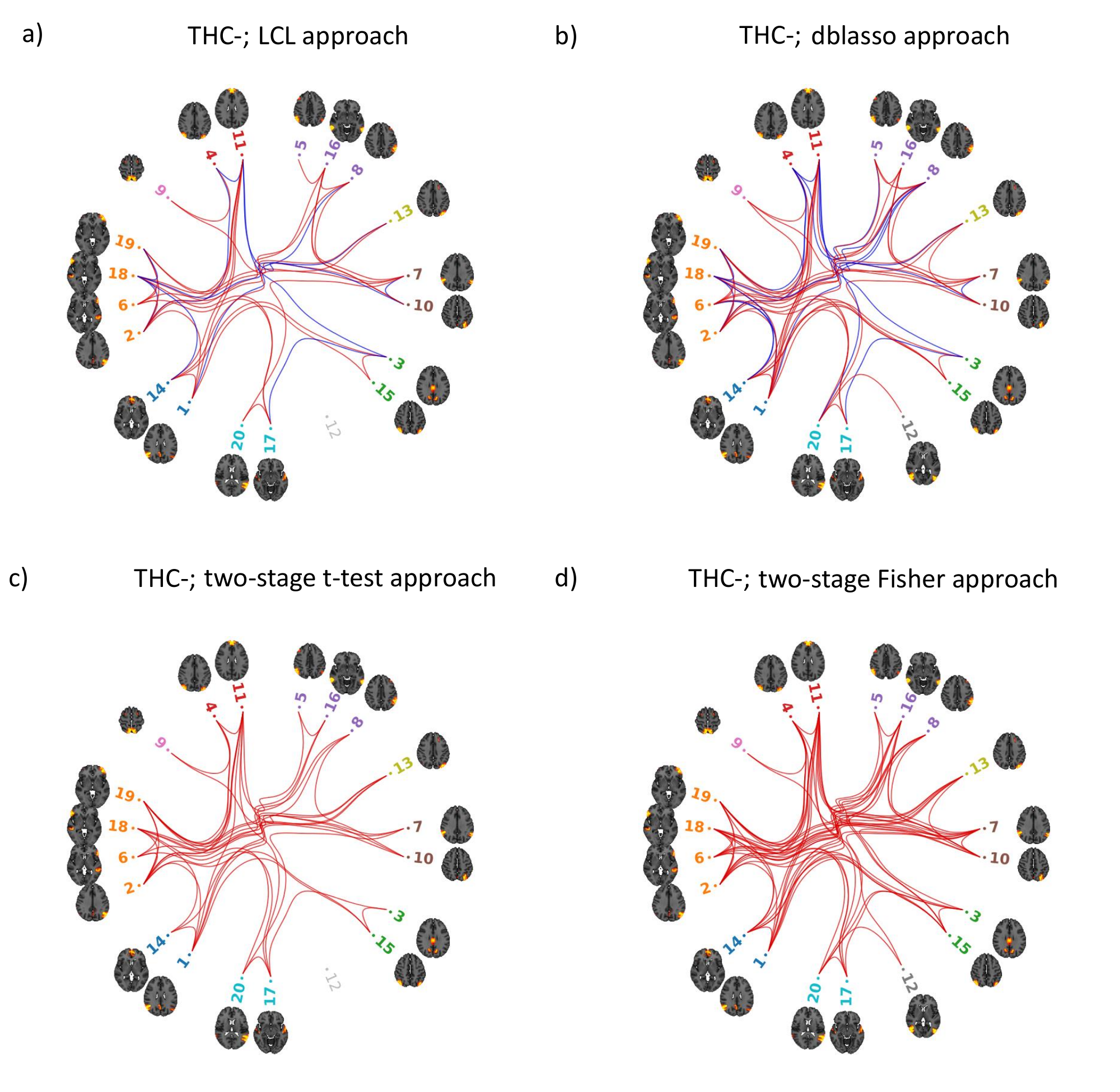}
    \caption{\footnotesize Brain sub-networks in the THC- group as estimated with the \textbf{a}: \emph{LCL} approach; \textbf{b}: \emph{dblasso} approach; \textbf{c}: \emph{two-stage t-test} approach; \textbf{d}: \emph{two-stage Fisher} approach. We present the results only for the nodes we selected in Figure~\ref{fig:data_beta}. Only edges with p-values $<1\times 10^{-13}$ are presented. Red edges represent positive edge strength, and blue edges represent negative edge strength. For two-stage approaches, the edges are tested without computing population-level edge estimates, so we plot all the edges in red.}
    \label{fig:data_compare}
\end{figure}

For comparison, we consider the \emph{LCL} method of \textcite{li2021inference}, and the de-biased LASSO method of  \textcite{zhang2014confidence} (referred to as \emph{dblasso}). 
We also apply two variants of two-stage approaches, using t-test (\emph{two-stage t-test}) or Fisher's method (\emph{two-stage Fisher}), to test the statistical significance of $\beta_{j \leftrightarrow k}$ for the THC- group. The \emph{dblasso} approach ignores the within-subject correlations and subject-level variations, yielding network strength estimates that are equivalent to applying our model with $a=0$ in the proxy matrices $\Sigma_a^{(j)}$. 

\begin{figure}[!htb]
    \centering
    \includegraphics[width=0.5\textwidth]{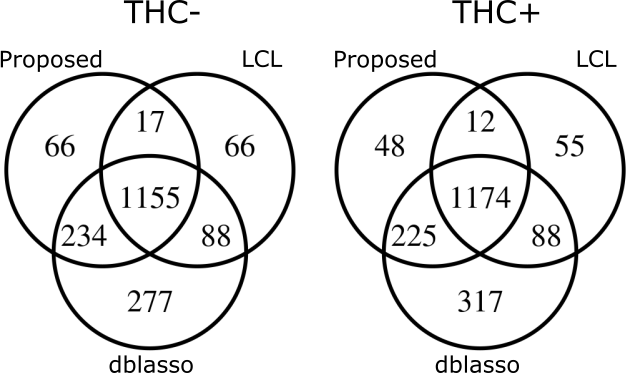}
    \caption{\footnotesize Venn diagrams of significant edges controlling for the family-wise error rate at 0.05 using Holm’s procedure, according to the proposed method, \textit{dblasso}, and \textit{LCL}. The diagram on the left corresponds to the THC- group, and the one on the right to the THC+ group. These results indicate that the proposed method not only avoids retaining some edges deemed significant by \textit{dblasso}, but also identifies additional edges that \textit{dblasso} does not detect.}
    \label{fig:venn_hits}
\end{figure}

Controlling for the family-wise error rate at 0.05 using Holm's procedure \parencite{holm1979simple}, the proposed method detects $1,472$ edges (7.4\%) in the THC- group and $1,459$ edges (7.3\%) in the THC+ group as significantly different from zero. The estimated brain networks in the two groups have $1,028$ edges in common. For better visualization, we only show the results for a sub-network with nodes that are associated with the default mode network, i.e., the brain regions that are active during passive rest and therefore most relevant to resting-state experiments. Also, we only plot the most significant edges (p-values $<1\times 10^{-13}$). The estimated sub-networks are shown in Figures~\ref{fig:data_beta}a-b. The proposed method also provides estimates of the variability of each edge. In Figures~\ref{fig:data_beta}c-d, we present the top 25\% of edges, selected based on the highest estimated variances, from those depicted in Figures~\ref{fig:data_beta}a-b.
The plots show that, for example, the connectivity between node 12 and node 20 presents high subject-level heterogeneity in both groups.

The link between node 3, the posterior cingulate cortex, and node 11, the medial prefrontal cortex, is deemed significant in the THC- but not in the THC+ group. These two areas serve as major hubs within the default mode network \parencite{deshpande2011instantaneous}. The disconnection of these regions has been linked to working memory deficits \parencite{whitfield2012default}. Even though the primary aim of this work is not to explore this research question, and considering that non-significant p-values merely suggest that the data do not provide sufficient evidence to reject the null hypothesis, these findings are consistent with research that has linked acute cannabis use to weakening abilities to maintain, manipulate, and recall information \parencite{heishman1997comparative}.

In terms of the number of discoveries, for the THC- group, \emph{dblasso} estimated 1,754 edges (8.8\%) as significant connections among 200 brain nodes and \emph{LCL} detected 1,326 edges (6.7\%); 2,648 edges (13.3\%) and 1,414 edges (7.1\%) were detected by \emph{two-stage t-test} and \emph{two-stage Fisher} approaches, respectively. The larger number of significant edges detected by \emph{dblasso} and \emph{two-stage t-test} is likely due to the inflated type-I error, as illustrated in simulations in Section~\ref{sec:simresults} when subject-specific heterogeneity is present in the data. We illustrate in Figure~\ref{fig:data_compare} the most significant edges (p-values $<1\times 10^{-13}$) detected by \emph{LCL}, \emph{dblasso}, \emph{two-stage t-test} and \emph{two-stage Fisher} among the previously selected brain nodes. While the overall patterns of connectivity are similar among these methods when compared to the proposed method, the stark difference in the number of detected edges by \emph{dblasso} and \emph{two-stage t-test} highlights the importance of modeling the heterogeneity in functional connectivity studies even when the population-level connectivity is of interest. In Figure~\ref{fig:venn_hits}, we show a Venn diagram of the significant edges for the main approaches considered here: the proposed method, \textit{LCL}, and \textit{dblasso}. We also note that with $200$ nodes and $120$ observations per subject, \emph{LCL} is not able to estimate the variances of the network edges. 

In terms of computation time, for the proposed method, selecting the optimal hyperparameter $a$ took approximately 50 seconds over a grid of 20 possible values, along with a grid of 20 possible values for $\lambda$. This grid was automatically extended to 40 values if the optimal $\lambda$ fell at the boundary of the initial grid. Once the optimal value of $a$ was selected, estimation and inference for all 199 edges from a fixed node $j$ required approximately 8 minutes. Note that, for each node $j$, the proposed inference procedure requires solving one LASSO regression for $\hat\beta_j$ and $199$ LASSO regressions for $\hat\kappa_{j,k}$, one for each $k\neq j$. Variance component estimation required approximately 50 seconds. The computation time for LCL was comparable. The \textit{dblasso} approach took approximately 6 minutes using default parameters. A more extensive study of computational timing for these methods is provided in the Supplementary Materials.

%%%%%%%%%%%%%%%%%%%%%%%%%%%%%%%%%%

\section{Simulation studies}
\label{section:simulation}
\subsection{Simulation setting: Neighborhood-based regression}\label{sec:simsettings}
% In this section, we use simulations to assess the performance of the proposed method. 
We first consider data generated from the following doubly high-dimensional LMM:
\begin{equation}
y^i = X^i (\beta + \gamma_i) + \epsilon_i, \quad i=1, \dots, n \label{sim.model}.
\end{equation}
We assume
$X^i\mid\Sigma_X^i\sim
\mathrm{MN}_{m\times p}(0,I_m,\Sigma_X^i)$,
where $\mathrm{MN}_{m\times p}(M,U,V)$ denotes the matrix normal distribution with mean matrix $M$, row covariance matrix $U$, and column covariance matrix $V$. We further assume
$\gamma_i\sim\mathrm{N}(0,\diag(\psi))$ and
$\epsilon_i\sim\mathrm{N}(0,\sigma_e^2I_m)$, independently of each other and of $X^i$.
This setting corresponds to a nodewise regression for a single node $j$ in a GGM; therefore, we omit $j$ from the notation.

\begin{figure}[!htb]
    \centering
    \includegraphics[width=0.85\textwidth]{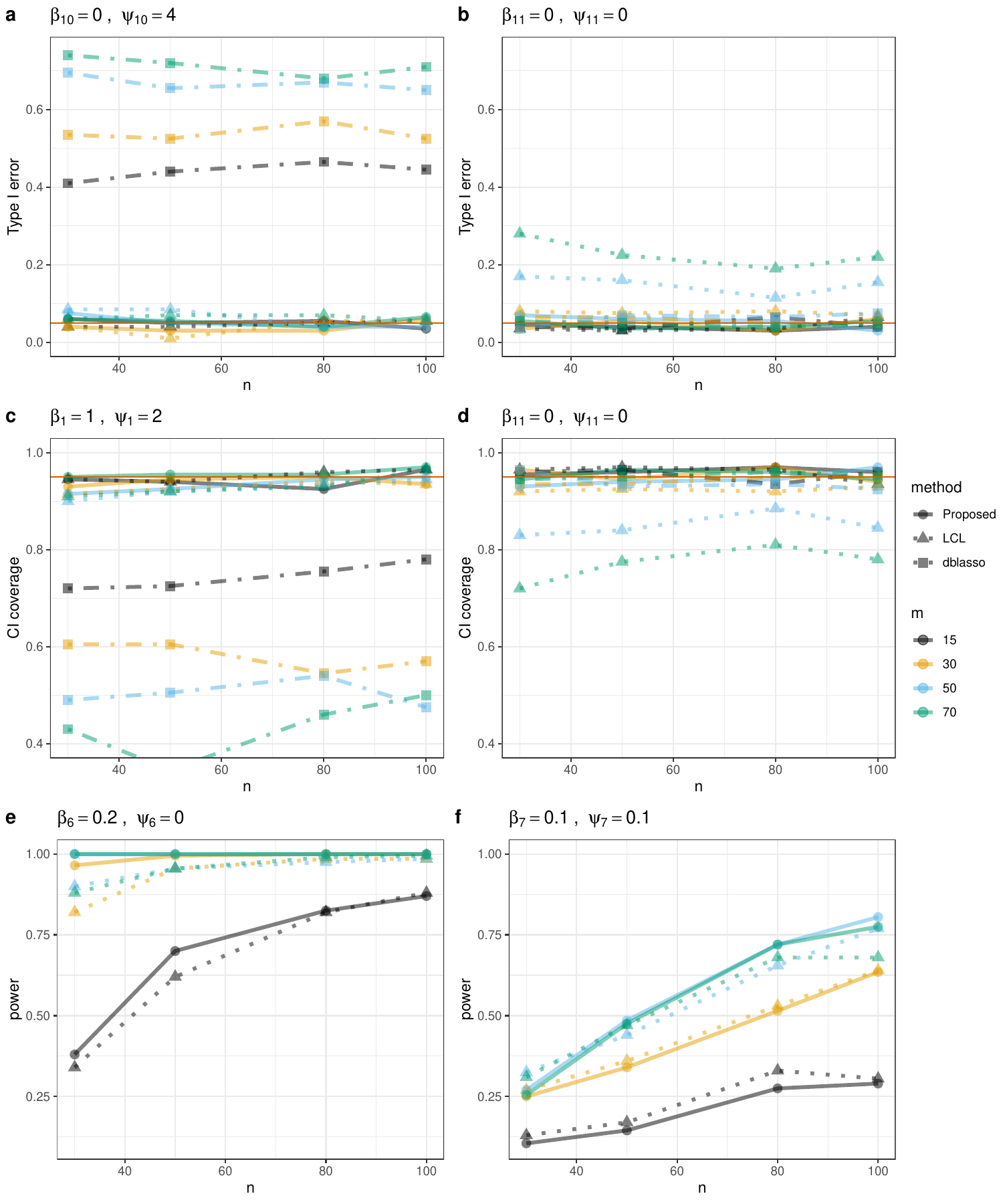}
    \caption{\footnotesize \textbf{a},\textbf{b}: Type-I error for testing $\beta_l$'s at the 0.05 significance level (0.05 marked by red solid line), for the proposed method, \emph{LCL} and \emph{dblasso}. 
    %The red solid line marks the position of 0.05. 
    \textbf{c},\textbf{d}: 95\% confidence interval coverage (0.95 marked by red solid line) for fixed effect coefficients, for the proposed method, \emph{LCL} and \emph{dblasso}. 
    %The red solid line marks the position of 0.95. 
    \textbf{e},\textbf{f}: Power for testing fixed effect coefficients at the 0.05 significance level, for the proposed method and \emph{LCL}. All results are computed at $p=20$ based on 200 Monte Carlo simulations, and are plotted separately for each method, each $m$, and against increasing $n$. The title of each subplot shows the true values of the targeted fixed effect coefficient $\beta_l$, and the corresponding random effect variance component $\psi_l$. 
    }
    \label{fig:1}
\end{figure}

Given the limitations of two-stage approaches, we compare our approach with the \emph{LCL} method of \textcite{li2021inference}, and the de-biased LASSO (referred to as \emph{dblasso}) method of  \textcite{zhang2014confidence}. As noted earlier, \emph{dblasso} serves as a benchmark approach that ignores the subject-level heterogeneity among observations. 
We compare the methods in terms of the total mean squared error (total MSE) for the estimates of all $\beta_l$ coefficients, the power to detect non-zero coefficients, the type-I error for testing zero coefficients at 5\% significance level, and the coverage of 95\% confidence intervals. We also compare the method proposed in Appendix~\ref{section:extension.vc} for estimation of the variance components, $\theta=(\psi, \sigma_e^2)$, with the method of \emph{LCL} by assessing the MSE of the noise variance $\sigma_e^2$, the total MSE for estimates of all $\psi_l$, and the selection consistency of the non-zero random effect variance components $\psi_{S_\psi}$. The selection consistency performance is evaluated via the Matthews correlation coefficient \parencite[MCC;][]{matthews1975comparison}, which summarizes true positive and false positive rates, with higher values indicating more accurate identification of the non-zero variance components. %To estimate the variance components, we equally split the sample into three sub-samples and followed the procedures in Section~\ref{section:extension.vc}.

We generated data from the doubly high-dimensional linear mixed model specified in \eqref{sim.model} with $n\in\{30,50,80,100\}$ and $m \in\{15, 30, 50, 70\}$. For each combination of $(n, m)$, we set $p \in \{20,60\}$ and conducted 200 independent Monte Carlo simulations. We set $(\beta_1^*, \beta^*_2, \beta^*_6, \beta^*_7, \beta^*_9) = (1,0.5,0.2,0.1,0.05)$ and $\beta^*_l = 0$ for the remaining $l$'s. The random effects $\gamma_i$ and the noise terms $\epsilon_i$ were independent realizations of two multivariate normal distributions $\mathrm{N}(0, \diag(\psi))$ and $\mathrm{N}(0, \sigma_e^2 I_m)$, respectively. The true values of the variance components $(\psi^*, \sigma_e^{2,*})$ are set as follows: $(\psi^*_1, \psi^*_4, \psi^*_7, \psi^*_9, \psi^*_{10},$ $\psi^*_{12}, \psi^*_{16}, \psi^*_{20})= (2, 2, 0.1, 0.1, 4, 0.1, 2, 0.1)$, with the remaining $\psi^*_l$'s set to 0, and $\sigma_e^{2,*} =1$. The fixed effect design matrices $X^i$ were independent realizations of a matrix normal distribution $X^i \mid \Sigma_X^i \sim \mathrm{MN}_{m \times p}(0, I_m, \Sigma_X^i)$. To generate $\Sigma_X^i$, we first generated a population-level covariance matrix $\Sigma_X$, which was set as a sparse random matrix with diagonal entries $1$ and off-diagonal entries drawn independently from a mixture distribution: each entry was either set to zero with probability $0.8$, or was drawn from a uniform distribution $\mathrm{Unif}(-0.5, 0.5)$. This choice was motivated by the nature of sparsely correlated brain networks. We repeated this procedure if the generated $\Sigma_X$ was not positive definite. Each $\Sigma_X^i$ was then generated as a perturbed version of $\Sigma_X$ by (i) determining varying off-diagonal entries of $\Sigma_X^i$ by drawing a Bernoulli random variable with success probability 0.2; and (ii) generating variations by adding a mean zero normal perturbation with standard deviation 0.1. To ensure symmetry, only the entries in the upper-triangular of $\Sigma_X^i$ were considered as candidates for perturbation. We repeated the above two steps if the generated $\Sigma_X^i$ was not positive definite. The resulting matrices $\Sigma^i_X$ represent subject-level heterogeneity in the brain networks.

\begin{figure}[!htb]
    \centering
    \includegraphics[width=0.85\textwidth]{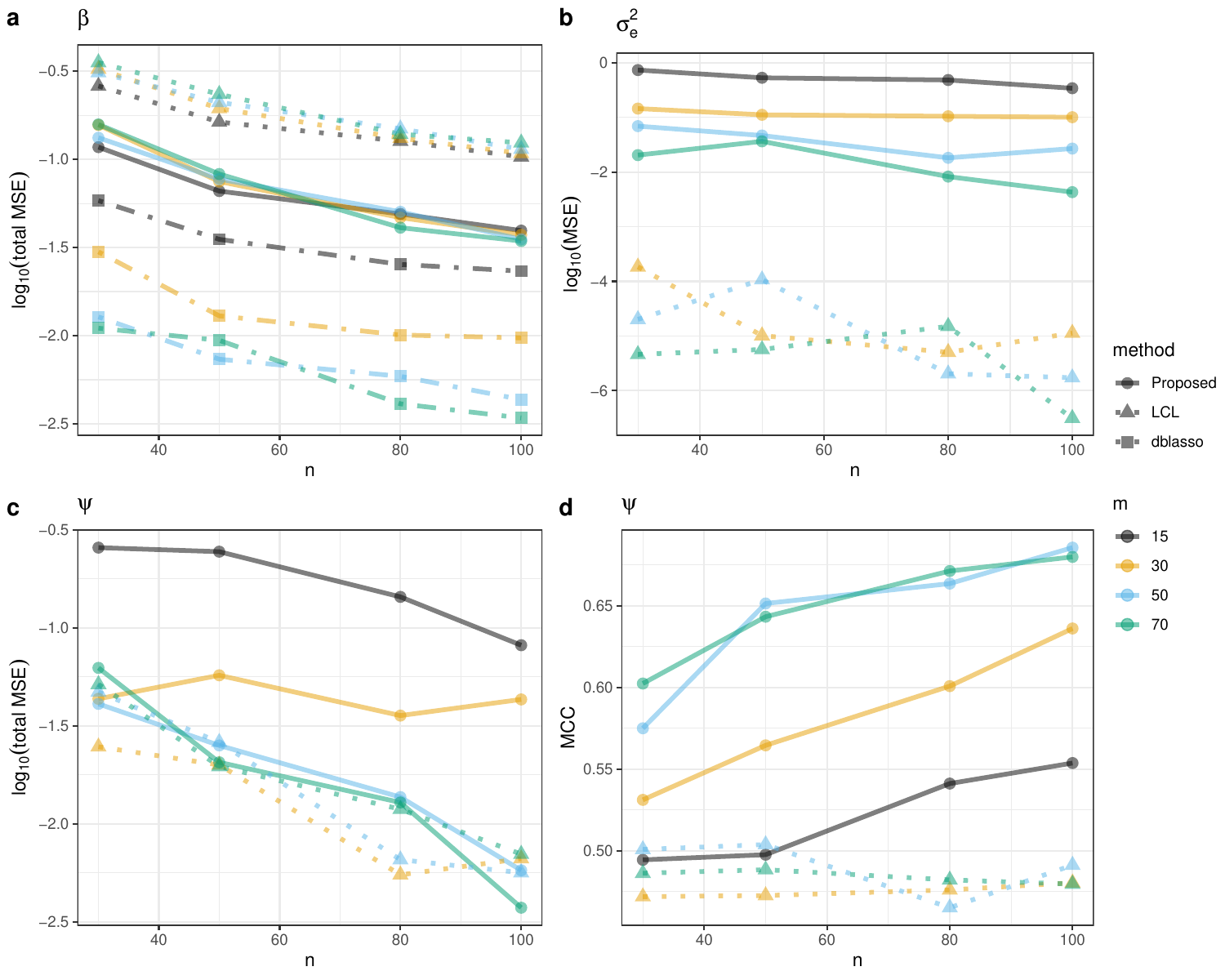}
    \caption{\footnotesize \textbf{a}: Total MSE on the $\log_{10}$ scale for estimating the fixed effect coefficients $\beta$, for the proposed method, \emph{LCL} and \emph{dblasso}. \textbf{b}: MSE on the $\log_{10}$ scale for estimating the noise variance $\sigma^2_e$, for the proposed method and \emph{LCL}. \textbf{c}: Total MSE on the $\log_{10}$ scale for estimating the random effect variance components $\psi$, for the proposed method and \emph{LCL}. \textbf{d}: Average MCC for identifying non-zero variance components for the proposed method and \emph{LCL}, where we use non-zero estimates to select non-zero variance components. Values are computed at $p=20$ based on 200 Monte Carlo simulations, are plotted separately for each method, each $m$, and are plotted against increasing $n$.}
    \label{fig:2}
\end{figure}

We used cross-validation with MSE as the error criterion to select the tuning parameters $\lambda_a$ for $\hat\beta$, $\lambda_k$'s for $\hat\kappa_k$ and $\lambda_\theta$ for variance components.  
We used the R function \texttt{cv.glmnet} from the R package \texttt{glmnet}  \parencite[\textit{v4.1-3};][]{glmnet} to implement the cross-validation algorithm. %We followed the same cross-validation approach to select the tuning parameter $\lambda_\sigma$ when estimating the variance components. 
The constant $a$ in the proxy matrix is also viewed as a tuning parameter. We followed the approach described in \textcite{li2021inference} to select an optimal value of $a$ via cross-validation: for each candidate value of $a$, we let the algorithm select the values for the rest of the tuning parameters, and used cross-validation based on MSE to select an optimal $a$. We present the results based on the optimal $a$. We used the authors' publicly available R code for \emph{LCL} \parencite{li2021inference}, and used the \texttt{hdi} R package \parencite{hdi} for \emph{dblasso}.

\subsection{Simulation results: Neighborhood-based regression}\label{sec:simresults}
Results for inference on fixed effect parameters $\beta$ are summarized in Figure~\ref{fig:1}. The proposed method controls the type-I error rate at the nominal level, whereas \emph{LCL} and \emph{dblasso} show inflated type-I errors in various settings (Figure~\ref{fig:1}a, \ref{fig:1}b): when the random effect variance is large ($\psi_{10}=4$ for $\beta_{10}$), tests by \emph{dblasso} have type-I errors as high as 0.74; \emph{LCL}'s type-I error reaches 0.28 when the random effect variance is zero (for $\beta_{11}$); both methods show higher type-I errors with increasing $m$.

The confidence intervals constructed by the proposed method provide good coverage, while those constructed by \emph{LCL} and \emph{dblasso} show poor coverage in some settings (Figure~\ref{fig:1}c, \ref{fig:1}d). The pattern of the confidence interval coverage is similar to the pattern of the type-I error: \emph{LCL} has coverage lower than 0.81 for $\beta_{11}$ at $p=20$, $m=70$, and has insufficient coverage for covariates with zero random effect variance when $m$ is large (Supplementary Table~S2); \emph{dblasso}'s coverage is always lower than 0.8 for some $\beta_l$'s (Figure~\ref{fig:1}c); both methods have worse coverage with increasing $m$. Since \emph{dblasso} has poor confidence interval coverage and highly inflated type-I error, we do not include it in the power comparison for testing $\beta_l$'s. The proposed method in general has comparable power against \emph{LCL} (Figure~\ref{fig:1}e, \ref{fig:1}f).

In terms of estimating the fixed effect coefficients, the proposed method always has smaller total MSE than \emph{LCL} (Figure~\ref{fig:2}a). Fixed effect coefficient estimates by \emph{dblasso} always have the smallest total MSE. Since \emph{dblasso} does not provide variance component estimates, it is not included for comparison in Figure~\ref{fig:2}.
When $m < p$, \emph{LCL} is not able to provide variance component estimates, while our method provides sensible estimates in all settings (Figure~\ref{fig:2}b, \ref{fig:2}c). When $m>p$, our method's estimate of $\sigma_e^2$ has higher MSE than \emph{LCL}'s (Figure~\ref{fig:2}b), but its total MSE for the random effect variance components is comparable to \emph{LCL}  (Figure~\ref{fig:2}c). 

We also examined the selection accuracy of the random effect variance components evaluated by MCC. The proposed method selects the variance components much more accurately than \emph{LCL}, and the accuracy improves with increasing $n$ and $m$ (Figure~\ref{fig:2}d). \emph{LCL} has poor selection accuracy, and is less accurate with increasing $n$. Examining the proportion of simulations that yield non-zero estimates for each variance component reveals that \emph{LCL} overselects zero variance components (Supplementary Figure~S2).

\begin{figure}[!htb]
    \centering
    \includegraphics[width=1\textwidth]{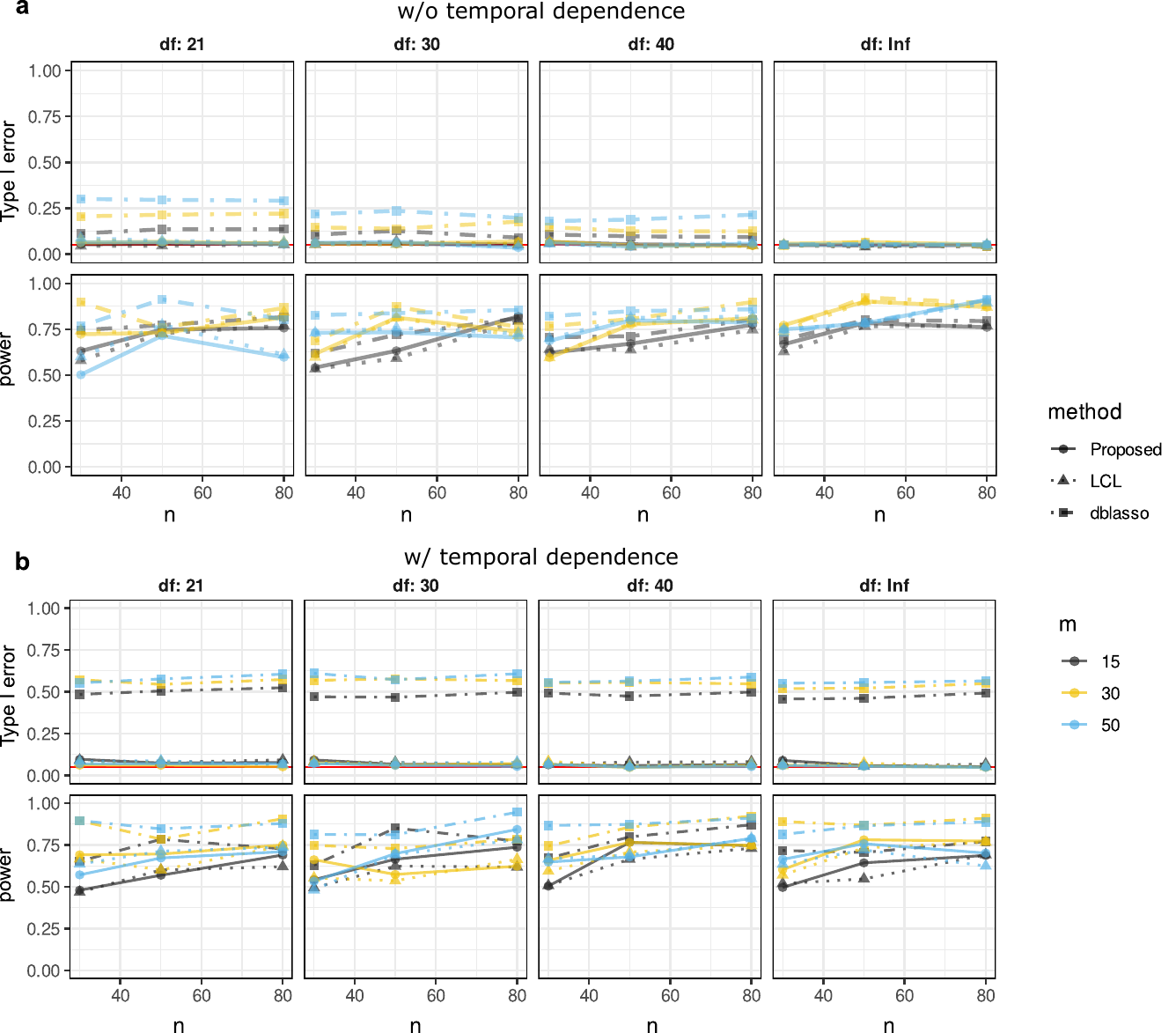}
    \caption{\footnotesize \textbf{a:} Type-I error and power for testing off-diagonal entries of the population precision matrix $\Sigma^{-1}$ at the 0.05 significance level (0.05 indicated by the red solid line) for the proposed method, \emph{LCL}, and \emph{dblasso}, when subject-specific covariances are generated from a Wishart distribution with $\mathrm{df}$ degrees of freedom, and no within-subject temporal dependence is introduced. \textbf{b:} Type-I error and power for testing off-diagonal entries of $\Sigma^{-1}$ at the 0.05 significance level for the proposed method, \emph{LCL}, and \emph{dblasso}, when subject-specific covariances are generated from a Wishart distribution and within-subject temporal dependence is introduced using a Toeplitz covariance matrix corresponding to an autoregressive process of order 2 with parameters 0.6 and 0.3. All results are computed with $p = 21$ based on 100 Monte Carlo simulations and are plotted separately for each method, each $m$, and across increasing values of $n$.}
    \label{fig:Wishart}
\end{figure}

In terms of computational time, the proposed method is slightly slower than \emph{LCL} for fixed effect estimation and inference, but is much faster for variance component estimation. The differences become more pronounced for large $p$ (Supplementary Figure~S1). Results for $p=60$ are presented in the Supplementary Materials. When $p=60$, the type-I error of \emph{LCL} improves for small $m$, but is still as high as 0.23 for large $m$ values; \emph{dblasso} has high type-I error regardless of $p$, which is not surprising. Interestingly, \emph{dblasso} often fails when testing covariates with non-zero random effect variances, and \emph{LCL} often fails when testing covariates with zero random effect variances, especially when $m$ is large.

% \subsection{Simulation Setting: Graph inference}\label{section:simulation_graph}

\subsection{Sensitivity analysis: Graph inference with misspecified models}\label{section:simulation_graph}

To demonstrate the robustness of our approach against a potentially misspecified generative model, in this section we test our method on data generated from a joint random graph model, which cannot be trivially mapped to the conditionally specified doubly high-dimensional LMM in equation (\ref{model.def.2}). Specifically, we assume that each $Y^i \in \mathbb{R}^{m \times p}$ is an independent realization from a matrix-normal distribution $Y^i \mid \Sigma_T, \Sigma_Y^i \sim \mathrm{MN}_{m \times p}(0, \Sigma_T, \Sigma_Y^i)$. We further assume that the covariance matrix $\Sigma_Y^i$ follows a Wishart distribution $\Sigma_Y^i \sim \mathcal{W}_p(\Sigma/\mathrm{df}, \mathrm{df})$, with mean matrix $\Sigma \in \mathbb{R}^{p \times p}$ and degrees of freedom $\mathrm{df}$. Recall that smaller degrees of freedom correspond to greater variability in the sampled $\Sigma_Y^i$. We assume that $\Omega = \Sigma^{-1}$ is sparse and generated using the same sparse-matrix construction described for $\Sigma_X$. The zero pattern of $\Omega$ is our target of inference, as it characterize the structure of the population-level precision matrix. The covariance matrix $\Sigma_T$ encodes dependence across observations, such as temporal dependence in the case of fMRI data. We consider two settings: (i) setting $\Sigma_T$ to the identity matrix, and (ii) setting it to a Toeplitz matrix corresponding to an AR model of order 2 with parameters 0.6 and 0.3. We generated data from the described model with $n \in \{30, 50, 80\}$, $m \in \{15, 30, 50\}$, and we set $p = 21$. 

Next, for a fixed node, we apply our proposed neighborhood regression mixed-model approach by treating the observations for that node as the outcome and the observations for all other nodes as both fixed and random effect covariates. The proposed method, baseline methods, and hyperparameter selection were implemented as described in the neighborhood-based regression simulation setting in Section~\ref{sec:simsettings}.

% \subsection{Simulation results: Graph inference}
The results in Figure~\ref{fig:Wishart} show that the proposed mixed effect graphical model achieves satisfactory control of the type-I error rate at the nominal level in both settings, with dependent and independent within-subject observations, and it maintains reasonable power for detecting non-zero coefficients. The \emph{dblasso} approach exhibits inflated type-I error, except in the case $\mathrm{df} = \infty$, which corresponds to no subject-specific variability in connectivity, i.e., $\Sigma_Y^1 = \cdots = \Sigma_Y^n = \Sigma$. When temporal dependence exists, a slight deviation from the nominal type-I error is observed for smaller sample sizes ($n=30$) and lower degrees of freedom of the Wishart distribution (i.e., high inter-subject variability). These results suggest that the proposed method is suitable for identifying non-zero direct edges at the population level in the presence of both inter-subject variability in connectivity and within-subject temporal dependence.

%%%%%%%%%%%%%%%%%%%%%%%%%%%%%%%%%%%%%%%%%%%%%%%%%%%%%
\section{Discussion}
\label{section:discussion}
Motivated by the problem of inferring the structure of population-level functional connectivity networks from multi-subject fMRI data, we proposed a mixed effect GGM. As shown in our numerical studies, ignoring the subject-level variations in GGMs can result in highly-inflated type-I errors and false discoveries, even when the population-level network is of main interest. 

Broadly speaking, there are two approaches to addressing subject-specific heterogeneity in functional connectivity. One is to explain this heterogeneity using subject-level features (e.g., behavioral or cognitive phenotypes) within a regression framework. Examples of such approaches include  \textcite{zhao2021principal, zhao2021covariate, desai2024connectivity}. These can identify statistically significant associations, but they typically explain only a small proportion of the observed variability, reflecting the fact that fMRI connectivity is extremely noisy, as it is derived from indirect measures of neuronal activity. However, it may also indicate that a continuum of connectivity configurations can give rise to similar behavioral or cognitive phenotypes (hence, the small proportion of variability explained). Another option to address this is to acknowledge such variability through the use of random effects, thus ensuring the statistical validity of the inferential framework for the population connectivity. This perspective aligns more closely with the approach we propose in this work. Moreover, our theoretical analysis avoids imposing strong distributional assumptions on inter-subject variation in functional connectivity and accommodates a broad spectrum of inter-subject variation patterns in functional connectivity.

The above two approaches are not mutually exclusive. Our framework naturally lends itself to extensions in which subject-level features are incorporated to explicitly model part of the subject-specific variability, as in \textcite{desai2024connectivity}, while residual unexplained variation is captured by random effects. Exploring such extensions represents an interesting direction for future work.

In summary, our framework provides rigorous inference of brain connectivity networks in the presence of subject heterogeneity for multi-subject studies. It facilitates a systematic exploration of population-level brain network structures, enhancing our understanding of the function associated with various brain regions. In our application, the proposed method is used to descriptively compare brain connectivity patterns between groups defined by THC screening status.

\subsection*{Funding}
AS was partially supported by National Institutes of Health (NIH) grants P30-AG066509 and RF1-AG090462. EL was partially supported by NIH grant R03-EB033001 and National Science Foundation (NSF) grant DMS-2210064. The content is solely the responsibility of the authors and does not necessarily represent the official views of the NIH, NSF, or any other funding agency.

%\clearpage
\appendix
\appendixpage

\section{Variance component estimation}
\label{section:extension.vc}

Estimating the variance components in doubly high-dimensional LMM settings introduces unique challenges that have not been addressed by existing approaches. In particular, the method of \textcite{li2018doubly} assumes bounded $m$, in order to use a Cholesky decomposition for estimating the random effect covariance matrix. The sample-splitting approach of \textcite{li2021inference} requires $m>q$ for $q$ random effect covariates, which restricts its applicability in high dimensions. We extend the sample-splitting approach to doubly high-dimensional LMM settings, and propose a penalized moment-based estimator for selecting and estimating the variance components. In particular, we allow for $m$ to be smaller than $q$. To simplify the problem, in the following, we assume that the noise terms are independent and identically distributed within each subject's observations, i.e., $R^i = \sigma_e^2 I_m$. Moreover, we assume $\Psi$ is a diagonal matrix satisfying Assumption~\ref{main.as.A.add}.\ref{main.as.A.3} in Section~\ref{section:theory}, such that $\Psi = \diag(\psi)$. To simplify the notation and to broaden the scope of the framework, we will present the proposed estimators and the theoretical results under a more general LMM formulation:
\begin{align}
    y^i = X^i \beta + Z^i \gamma_i + \epsilon_i, \quad i=1, \dots, n, \label{main.model:general_LMM}
\end{align}
where each $y^i \in \R^{m}$ is the observation vector, $X^i \in \R^{m\times p}$ and $Z^i \in \R^{m \times q}$ are the design matrices with $ X^i \mid \Sigma_X^i \sim \mathrm{MN}_{m \times p}(0,I_m, \Sigma_X^i)$ and $Z^i \mid \Sigma_Z^i \sim \mathrm{MN}_{m \times q}(0, I_m, \Sigma_Z^i)$. Conditional on $X^i$, the random effect coefficients $\gamma_i \in \R^{q}$ and the noise term $\epsilon_i \in \R^{m}$ are independent with covariance $\Psi$ and $R^i$, respectively. The fixed effect coefficient vector $\beta$ has support $S$ with cardinality $s$. 

It is known that the variance components become non-identifiable if the random effect covariance matrix is proportional to the noise covariance \parencite{wang2013identifiability}. In our case, this would happen if $Z^i \Psi (Z^i)^\top = c \sigma_e^2 I_m$, for some constant $c>0$. However, we assume that the diagonal of $\Psi$ is sparse and that $\sigma_{\min}\left(Z^i \Psi (Z^i)^\top\right) = 0$ whereas $\sigma_{\min}(\sigma_e^2 I_m) = \sigma_e^2 >0$, ensuring identifiability.

%Denote the true values of $\theta$ by $\theta^* = (\psi^*, (\sigma_e^*)^2)$. 
Let $\theta = (\psi, \sigma_e^2)$ denote the variance components. To estimate $\theta$, we adopt a sample-splitting approach, and partition the $n$ subjects into three sub-samples of similar sizes with index set $S_k$, $k=1, 2, 3$, such that $n_k = |S_k|$ and $n_1 \asymp n_2 \asymp n_3$. We start by using the first $n_1$ subjects to estimate the fixed effect parameters $\beta$, denoting the estimates as $\hat\beta$. We then use the second sub-sample with $n_2$ subjects to estimate the vector $\psi$. Denoting $\hat r_i = y^i - X^i\hat\beta$, $i \in S_2$, we define the penalized moment-based estimator for $\psi$ with tuning parameter $\lambda_\theta$, as:
\begin{align}
\hat{\psi} & = \argmin_{\psi \in \R^{q}}\sum_{i \in S_2} \left\|\hat r_i \hat r_i^\top - \diag(\hat r_i)\diag(\hat r_i) -  Z^i \diag(\psi)  (Z^{i})^\top + \sum_{l=1}^{q} \psi_l \diag(Z_l^i)\diag(Z_l^i)\right\|_F^2 \nonumber \\
& \quad + \lambda_\theta\|\psi\|_1. \label{main.hat.psi.def}
\end{align}
The estimator for $\psi$ is constructed from estimates of the second moment of the residuals $r_i = y^i-X^i\beta^*$ by noticing that $\mathbb{E}\left[ \left(r_i r_i^\top \right)_{l,k} | Z^i\right] = \left(Z^i \Psi (Z^i)^\top\right)_{l,k}$ for all $l \neq k$. In the high-dimensional setting considered here, we adopt a LASSO regularization to guarantee the consistency of the estimator and simultaneously perform variable selection. Negative entries of $\hat{\psi}$ are truncated to zero.

Finally, we use the third sub-sample with $n_3$ subjects to estimate the noise variance $\sigma_{e}^2$. To this end, we propose a simple moment-based estimator defined as 
\begin{align}
    \hat\sigma_{e}^2 = \frac{1}{n_3m} \sum_{i \in S_3}\Tr\left(\hat r_i \hat r_i^\top - Z^i \diag(\hat\psi) (Z^i)^\top\right), \label{main.noise.vc.def}
\end{align}
where, for $i\in S_3$, $\hat r_i = y^i - X^i \hat\beta$, with $\hat\beta$ estimated from $S_1$, and $\hat\psi$ is computed from $S_2$. 

\section{Theoretical results}\label{section:theory}
\subsection{Additional notation}
\label{subsection:notation}
We denote by $\|A\|_2$ the matrix norm of $A$, which is the maximum singular value of $A$. The Frobenius norm of $A$ is denoted by $\|A\|_F$, which is equal to $\sqrt{\tr(A A^\top)}$ with $\tr(A)$ denoting the trace of $A$. We use $\sigma(A)$, $\sigma_{\min}(A)$ and $\sigma_{\max}(A)$ to represent the singular values, the minimum singular value, and the maximum singular value of $A$. 

A random variable $U$ is sub-Gaussian with parameter $u$ if $\forall\ t \in \R$, $\E\left(e^{tU}\right) \leq \exp(ut^2/2)$. We define $\SG(u)$ as the class of all sub-Gaussian random variables with mean 0 and parameter $u$. A random vector $V$ is sub-Gaussian with parameter $v$ if for any vector $x$ with $\|x\|_2=1$, we have $x^\top V \in \SG(v)$. We denote the class of such sub-Gaussian random vectors by $\SGV(v)$.

For two scalars $a$ and $b$, we write $a \asymp b$ if $c_1|b| \leq |a| \leq c_2 |b|$ for some positive constants $c_1$, $c_2$. We write $a \vee b = \max(a, b)$, and $a \wedge b = \min(a, b)$. We use $a = O(b)$ to indicate that $a \leq c_1 b$ for some constant $c_1>0$, and use $a=o(b)$ to mean that $a/b \xrightarrow{}0$. %Note that $a=O(b)$ includes the case when $a=o(b)$. 
Throughout the paper, we use $c,\ c_0,\ c_1, \ldots$ to represent positive constants, whose values may vary from line to line.

\subsection{Fixed effect estimation and inference}
To achieve consistent estimation for the doubly high-dimensional LMM in equation \eqref{model.def.2}, we assume the true $\beta_j^*$ coefficients are sparse, with support $S_j$ and cardinality $s_j = |S_j|\ll p$. Let $\Sigma^i$ be the $p\times p$ subject-specific covariance matrix for subject $i$. To simplify the theoretical analysis, we assume that, conditional on the covariance matrices $\Sigma^i$, the matrices $Y^i$ are independent and follow a matrix normal distribution $Y^i \mid \Sigma^i \sim \mathrm{MN}_{m \times p}(0,I_m,\Sigma^i)$, for $i=1, \dots, n$. This implies that within-subject observations are independent conditional on the subject-specific covariance matrix. 
%the rows of $Y^i$ are independent realizations of $N(0, \Sigma^i)$, for $i=1, \dots, n$. 
To allow for subject-level heterogeneity, we assume the subject-specific covariance matrices $\Sigma^i$ are random, and we denote by $\Sigma$ the corresponding population-level covariance matrix. Conditional on the matrix $Y^i_{-j}$, the random effect coefficients $\gamma^i_j$ and the noise vectors $\epsilon^i_j$ are independent with mean zero and covariance $\Psi_j \in \R^{(p-1)\times (p-1)}$ and $R^i_j \in \R^{m \times m}$, respectively. We assume $\gamma^i_j \in \SGV(c_1 \|\Psi_j\|_2)$ and $\epsilon^i_j \in \SGV(c_2 \|R^i_j\|_2)$, which implies that $Y^i_j - Y^i_{-j}\beta_j | Y^i_{-j}$ is conditionally sub-Gaussian. 
In this section, we show that, under mild conditions, the proposed LASSO estimator $\hat\beta_j$ in \eqref{eq:fixed_effect_est_definition} is consistent, and the proposed de-biased LASSO estimator $\hat\beta^{(\mathrm{db})}_{j,k}$ in \eqref{eq:dblasso} is asymptotically normal.
We first state the assumptions under which we establish the consistency of $\hat\beta_j$ defined in \eqref{eq:fixed_effect_est_definition}. We use $c, c_0, c_1, \dots$ to denote generic positive constants whose values may vary line by line.

\begin{assumption}
\label{main.as.A}
\begin{enumerate}
    \item[]
    \item \label{main.as.A.1} The number of observations per subject $m$ and the number of covariates $p$ satisfy $p>c_0m$ for some constant $c_0>1$, and $p > c_1$ for some suitably large constant $c_1>0$. Moreover, $p\log(p)/(mn) = o(1)$.
    \item \label{main.as.A.2} $\forall \ i, j$, $\sigma(\Sigma) \asymp \sigma(R^i_j) \asymp \|\Psi_j \|_2 \asymp 1$, and $\|\Sigma^i - \Sigma\|_2 \leq  \sigma_{\min}(\Sigma) -c_2$ for some constant $c_2 >0$. 
    % \item \label{as.A.4} 
    % \begin{align*}
    %     \begin{cases}
    %     \frac{s^2q\log(p)}{mn} =o(1) &, \text{ when } q>c_0m\\
    %     \frac{s^2 \log^2(n)\log(p)}{n} =o(1) &, \text{ when } q>c_0m \text{ and Assumption~\ref{as.A.add}.\ref{as.A.3} holds}\\
    %     \frac{s^2\log(p)}{n} =o(1) &, \text{ when } m>c_0q \text{ and } p=q\\
    %     \frac{s^2m\log(p)}{n} =o(1) &, \text{ when } m>c_0q \text{ and } p>q
    %     \end{cases}
    % \end{align*}
\end{enumerate}
\end{assumption}

In Assumption~\ref{main.as.A}.\ref{main.as.A.1}, we restrict the growth rate of $p$ relative to $m$ and $n$. Moreover, we assume $p/m$ is no smaller than a positive constant $c_0>1$, which is the situation most relevant to our application. In the Supplementary Materials, we discuss the case when $p/m \leq 1/c_0$ for some constant $c_0>1$. % The case $p/m \xrightarrow{} 1$ brings about complications, and we leave them to future investigation. 

By Weyl's theorem, Assumption~\ref{main.as.A}.\ref{main.as.A.2} implies $\sigma(\Sigma^i) \asymp 1$ for all $i=1, \dots, n$ \parencite{bhatia2007perturbation}, requiring the singular values of the covariance matrices $\Sigma$, $\Sigma^i$, and $R^i_j$ to be both lower and upper bounded by positive constants.

To establish the consistency of $\hat\beta_j$, we only require an upper bound on the singular values
of the covariance matrices $\Psi_j$ in Assumption~\ref{main.as.A}.\ref{main.as.A.2}. However, later we also require one of the following assumptions on $\Psi_j$ to hold in order to establish the asymptotic normality of the de-biased estimator; we state these assumptions here for convenience.

\begin{assumption}
\label{main.as.A.add}
\begin{enumerate}
    \item[]
    \item \label{main.as.A.3} (Diagonal structure): $\forall \ j$, $\Psi_j = \diag(\psi_j)$ for a vector $\psi_j \in \R^{p-1}$. The support of $\psi_j$ is $S_{\psi_j}$ with cardinality $s_{\psi_j} < c_1 (m \wedge n)$ for some constant $c_1>0$. Moreover, $\min(\psi_{S_{\psi_j}}) \asymp \max(\psi_{S_{\psi_j}}) \asymp 1$.
    \item \label{main.as.A.3.2} (Bounded eigenvalues): $\forall \ j$, $\sigma_{\min}(\Psi_j) \asymp 1$.
\end{enumerate}
\end{assumption}

Assumption~\ref{main.as.A.add}.\ref{main.as.A.3} allows us to cover the settings when the minimum singular values of the covariance matrices $\Psi_j$ are not bounded away from zero. In such a case, we consider a sparse diagonal structure for the matrices $\Psi_j$, under which our results hold with slightly different sample size assumptions.

The next result establishes the theoretical properties of the proposed estimator $\hat \beta_j$ in \eqref{eq:fixed_effect_est_definition}. 

\begin{theorem}[Fixed effect estimator error bounds]
\label{main.thm:1}
Suppose  Assumption~\ref{main.as.A}.\ref{main.as.A.1} and Assumption~\ref{main.as.A}.\ref{main.as.A.2} hold and that $\lambda_{a,j} = c_1\sqrt{p\log(p)/(nm)}$ for a suitably large $c_1>0$. Then, there exists a constant $C>0$ such that, with probability at least $1-4\exp\{-cn\} -12\exp\{-c\log(n)\} - 3\exp\{-cmnp^{-1}\}$,
     \begin{align*}
         &\left\|\hat \beta_j -\beta_j^*\right\|_\delta \le C \left(s_j^{1/\delta}\sqrt{\frac{p\log(p)}{mn}}\right), \quad \delta \in \{1,2\}, \\
         & \left\| (\Sigma_a^{(j)})^{-1/2} Y_{-j} (\hat \beta_j -\beta_j^*)\right\|_2^2 \le C \left(s_j\log(p)\right).
     \end{align*}
If, in addition, Assumption~\ref{main.as.A.add}.\ref{main.as.A.3} holds, taking $\lambda_{a,j} = c_2\sqrt{\log(p)\log^2(n)/n}$ for a suitably large $c_2>0$, there exists a constant $C>0$ such that, with probability at least $1-4\exp\{-cn\} -12\exp\{-c\log(n)\} - 3\exp\{-cmnp^{-1}\}$,
     \begin{align*}
         &\left\|\hat \beta_j -\beta_j^*\right\|_\delta \le C \left(s_j^{1/\delta}\sqrt{\frac{\log^2(n)\log(p)}{n}}\right),\quad \delta\in\{1,2\},\\
         &\|(\Sigma_a^{(j)})^{-1/2} Y_{-j} (\hat \beta_j -\beta_j^*)\|_2^2 \le C \left(\frac{s_jm\log^2(n)\log(p)}{p}\right).
     \end{align*}

\end{theorem}

Theorem~\ref{main.thm:1} shows that the de-correlated LASSO estimator $\hat\beta_j$ achieves $\ell_1$, $\ell_2$ and prediction consistency. The proof of Theorem~\ref{main.thm:1} is based on the classical proof of the consistency of the LASSO estimator in regression settings \parencite{buhlmann2011statistics} with multiple key modifications to extend it to our setting. One crucial step is to show that the restricted eigenvalue condition \parencite{buhlmann2011statistics} holds for the matrix product $X^\top \left(aZZ^\top+I\right)^{-1} X$, where $X$ is the fixed effect design matrix and $Z$ is the random effect design matrix. Since in our setting the fixed and random effect design matrices are identical, we cannot rely on techniques such as those adopted in \textcite{li2021inference} where the conditional mean independence assumption between $X$ and $Z$ is used to remove the dependence on the matrix $Z$. Instead, we jointly study the contribution of both the fixed and random effect design matrices, and use results from random matrix theory \parencite{tropp2015introduction} to directly characterize the eigenvalue distribution of the relevant quantities. In particular, we obtain a set of tight bounds for functions of the non-zero singular values of the matrices $Y^i$ under Assumption~\ref{main.as.A}. This key result, which may be of independent interest, is presented in the Supplementary Materials.

Next, we state a simplified version of the assumptions under which we establish the asymptotic normality of the de-biased LASSO estimator $\hat\beta^{(\mathrm{db})}_{j,k}$ defined in \eqref{eq:dblasso}. In order to allow for such a simplification, we have made some additional mild assumptions such as $\sigma\left(G^{(j)}_k\right) \asymp p$ under Assumption \ref{main.as.A.add}.\ref{main.as.A.3.2}, $ c_1 \leq \sigma_{\min}\left(G^{(j)}_k\right) \leq \sigma_{\max}\left(G^{(j)}_k\right) \leq c_2 m$ under Assumption~\ref{main.as.A.add}.\ref{main.as.A.3}, and $|H_k^{(j)}| \asymp s_j \asymp 1$. Recall that model \eqref{model.def.2} implies that given $Y^i_{-j}$, the vector $Y^i_j - Y^i_{-j}\beta_j$ is sub-Gaussian with covariance matrix $\Sigma_\theta^{i,(j)} = Y^i_{-j} \Psi_j (Y^i_{-j})^\top + R^i_j$. If we were to assume that given $Y^i_{-\{j,k\}}$, the covariance matrix of $Y^i_k$ has a ``sandwich'' form akin to $\Sigma_\theta^{i,(j)}$, namely $G_k^{(j)} = Y^i_{-\{j,k\}} \Psi_{j,k} (Y^i_{-\{j,k\}})^\top + R^i_{j,k}$ for some matrix $\Psi_{j,k} \in \R^{(p-2)\times (p-2)}$ and $R^i_{j,k} \in \R^{m\times m}$, we could then characterize the rates of $\|G_k^{(j)}\|_2$ and $\sigma_{\min}\left(G_k^{(j)}\right)$ at the rates assumed above.  %(see Lemma~\ref{lemma:sigma-theta}). 

\begin{assumption}
\label{main.as.B}
\begin{enumerate}
    \item[]
    \item \label{main.as.B.1} $\forall \ j,k$, conditioning on the matrix $Y^i_{-\{j,k\}}$, the vector $Y^i_k - Y^i_{-\{j,k\}} \kappa_{k}^{(j),*}$ has mean zero and covariance $G^{(j)}_k$, and belongs to the class $\SGV( c_1\|G^{(j)}_k\|_2)$.
    %$\E(Y^i_k \mid Y^i_{-\{j,k\}})= Y^i_{-\{j,k\}} \kappa_{k}^{(j),*}$, $\mathrm{Var}(Y^i_k \mid Y^i_{-\{j, k\}}) = G^{(j)}_k$, and $Y^i_k - Y^i_{-\{j,k\}} \kappa_{k}^{(j),*} \mid Y^i_{-\{j,k\}} \in \SGV( c_1\|G^{(j)}_k\|_2)$. 
    The support of $\kappa_{k}^{(j),*}$,  $H^{(j)}_k$, has cardinality $|H^{(j)}_k|$ satisfying $\|\kappa_{k}^{(j),*}\|_1 \leq c_1|H^{(j)}_k|$. 

    \item \label{main.as.B.2}
    One of the following conditions holds: 
    \begin{itemize}
        \item Assumption~\ref{main.as.A.add}.\ref{main.as.A.3} holds, and $p^2\log(p) \ll nm^3$, $m s_{\psi_j} \log^7(n) \ll n$, $\log(n)\log(p) \ll n$. 
        \item Assumption~\ref{main.as.A.add}.\ref{main.as.A.3.2} holds, and $p\log(n)\log(p) \leq c_3mn$.
    \end{itemize}

    \end{enumerate}
    \end{assumption}

In Assumption~\ref{main.as.B}.\ref{main.as.B.1}, we specify an upper bound for $\|\kappa_k^{(j),*}\|_1$. This is not too restrictive, because given that the variance of each node is bounded (implied by Assumption~\ref{main.as.A}.\ref{main.as.A.2}), it is reasonable to expect that the coefficients $\kappa_k^{(j),*}$ are not too large in absolute value. In the case of no subject-level heterogeneity, that is $Y^i \sim \mathrm{N}(0, \Sigma)$ and $\kappa_k^{(j),*} = \left(\Sigma_{-\{j,k\}, -\{j,k\}}\right)^{-1} \Sigma_{-\{j,k\}, k}$. Thus, in this case $\kappa_k^{(j),*}$ corresponds to the scaled $k$th row of the precision matrix $\Omega_j = (\Sigma_{-j,-j})^{-1}$ for nodes in $V\backslash \{j\}$. Then, we can show that $\sum_{l=1, l\neq k}^{p-1} \left| (\kappa_k^{(j),*})_l \sqrt{(\Omega_j)_{k,k}} /\sqrt{(\Omega_j)_{l,l}} \right| \leq |H_k^{(j)}|$, and $\|\kappa_k^{(j),*}\|_1 \leq c_1 |H_k^{(j)}|$ is satisfied.

In Assumption~\ref{main.as.B}.\ref{main.as.B.1}, we also specify the conditional distribution of $Y^i_k$ given $Y^i_{-\{j,k\}}$. We do not assume a specific structure for the covariance matrix $G_k^{(j)}$, but only require mild bounds on the minimum and maximum eigenvalues. %(Appendix~\ref{S:B}). 

Assumption~\ref{main.as.B} indicates different sample size requirements according to different structures of $\Psi_j$ in Assumption~\ref{main.as.A.add}. %This is intuitive since the structure of the random effect covariance matrix, and the magnitude of the random effect variance impact how difficult it is to infer the corresponding fixed effect coefficients. 
Note that Assumption~\ref{main.as.A.add} can be relaxed at the cost of stricter sample size requirements: if we only assume $\|\Psi_j\|_2\leq c_1$ and drop Assumption~\ref{main.as.A.add}, the asymptotic normality property of $\hat\beta^{(\mathrm{db})}_{j,k}$ will hold with some additional sample size assumptions (see Supplementary Materials), which would restrict the growth rate of $p$ to be slower than $n$.

The next result establishes the asymptotic normality of the estimator $\hat\beta^{(\mathrm{db})}_{j,k}$ in \eqref{eq:dblasso}. 

\begin{theorem}
\label{main.thm:2}
Under Assumptions~\ref{main.as.A} and \ref{main.as.B}, we have that $\left(V_{j,k}\right)^{-1/2}\left(\hat\beta_{j,k}^{(\mathrm{db})} - \beta_{j,k}^*\right) = R_{j,k} + o_p(1)$, where $R_{j,k} \xrightarrow[]{d} \mathrm{N}(0,1)$ and
% \begin{align*}
%     &\frac{1}{\sqrt{V_{j,k}}}\left(\hat\beta_{j,k}^{(\mathrm{db})} - \beta_{j,k}^*\right) = R_{j,k} + o_p(1), \quad \text{where } R_{j,k} \xrightarrow[]{d} \mathrm{N}(0,1),
% \end{align*}
the variance $V_{j,k}$ is given by
\begin{equation*}
    V_{j,k} = \frac{ \sum_{i=1}^n (\hat w_{j,k}^i)^\top \left(\Sigma_b^{i,(j,k)}\right)^{-1/2} \Sigma^{i,(j)}_{\theta} \left(\Sigma_b^{i,(j,k)}\right)^{-1/2} \hat w^i_{j,k}}{\left| \sum_{i=1}^n (\hat w_{j,k}^i)^\top \left(\Sigma_b^{i,(j,k)}\right)^{-1/2} Y^i_{k} \right|^2}.
\end{equation*}
\end{theorem}

The proof of Theorem~\ref{main.thm:2} builds on the properties of the projection vectors $\hat\kappa_{j,k}$ and the orthogonal complement of the projection $\hat w^i_{j,k}$. Specifically, we show that $\left(V_{j,k}\right)^{-1/2}\left(\hat\beta_{j,k}^{(\mathrm{db})} - \beta_{j,k}^*\right)$ can be divided into two terms, where one term is $o_p(1)$ and the other term can be shown to be asymptotically normal, thanks to the Lyapunov central limit theorem. 

A novel feature of our results, compared to those in the literature, including \emph{LCL} \parencite{li2021inference}, is that we establish the unconditional asymptotic normality of $\hat\beta^{(\mathrm{db})}_{j,k}$. In contrast, other approaches establish the asymptotic normality only conditionally on the random effect design matrices, even when the design matrices are assumed to be random. This unconditional normality is crucial for our application to brain connectivity network inference, and to the best of our knowledge, this is the first attempt to characterize the unconditional asymptotic properties of the fixed effect coefficients estimator with random design matrices in high-dimensional LMMs. The lemmas required to prove Theorem~\ref{main.thm:2} are presented in the Supplementary Materials.

Finally, we show that the sandwich estimator $\hat V_{j,k}$, defined in \eqref{df.hat.v}, is a consistent estimator of $V_{j,k}$ under Assumption~\ref{main.as.C}. We apply the same simplifications as we did for Assumption~\ref{main.as.B}. 

 \begin{assumption}
 \label{main.as.C}
 One of the following conditions holds: 
 \begin{enumerate}
     \item Assumption~\ref{main.as.A.add}.\ref{main.as.A.3} holds and $m^2\log^5(n)\log^2(p) \ll n$.
     \item Assumption~\ref{main.as.A.add}.\ref{main.as.A.3.2} holds and 
     $\log^2(p)\log^4(n) \ll n$.
 \end{enumerate}
 \end{assumption}

\begin{theorem}
\label{main.lemma:Vj}
Under Assumptions~\ref{main.as.A}, \ref{main.as.B}, and \ref{main.as.C}, we have $\hat V_{j,k}/V_{j,k}  = 1 + o_p(1)$.
\end{theorem}

\subsection{Variance estimation}
In the general LMM setting introduced in Appendix~\ref{section:extension.vc}, we assume that $\gamma_i \in \SGV(c_1 \|\Psi\|_2)$, $\epsilon_i \in \SGV(c_2\|R^i\|_2)$, respectively. Moreover, the preceding assumptions are understood to apply with the node index $j$ suppressed. We gather additional assumptions needed to prove the consistency of the proposed variance component estimator $(\hat\psi, \hat\sigma_e^2)$ in Assumption~\ref{main.as.D} below. The true values are denoted by $\theta^* = (\psi^*, \sigma_e^{2, *})$.

\begin{assumption}
 \label{main.as.D}
 \begin{enumerate}
     \item[]
     \item \label{main.as.D.1} The vectors $\gamma_i$ and $\epsilon_i$ are normally distributed with mean zero and covariance matrices $\Psi$, $\sigma_e^2 I_m$, respectively. The covariance matrix $\Psi$ satisfies Assumption~\ref{main.as.A.add}.\ref{main.as.A.3}.
     \item \label{main.as.D.2} Letting $s_Z = \max_i\max_j \sum_{k=1, k\neq j}^{q} \bm{1}\left\{|(\Sigma_Z^i)_{j,k}| > C_Z \log(nq^2)/\sqrt{m} \right\}$ for a sufficiently large constant $C_Z>0$, we have\\ 
     $\sqrt{m} \gg \log(nq)$,
     $nm \gg \max \left\{ q^{3/2}\log(q) \log^2(n), \log(q) \log(n)\log(nq^2) \left(s_Z\sqrt{m} + q\log(nq^2)\right) \right\}$,
     $nm^3 \gg q^2\log(q)\log^2(n)$.
     \item \label{main.as.D.4} $\sqrt{n}m^2 \gg s s_\psi q \log(p) \log(q) \log(n)\log(nmq)$.
 \end{enumerate}
 
\end{assumption}

\begin{theorem}
\label{main.thm:Svc}
Under Assumption~\ref{main.as.A} and \ref{main.as.D}.\ref{main.as.D.1}--\ref{main.as.D}.\ref{main.as.D.2}, with probability at least $1 -c_1\exp\{-c\log(nq)\}-c_2\exp\{-cn\}-c_3\exp\{-c\log(n)\}- c_4\exp\{-cmn/q\} - c_5\exp\{-c\log(q)\}$, we have 
\begin{align*}
    \|\hat\psi- \psi^*\|_\delta & \leq s_{\psi}^{1/\delta}  \frac{q\log(n)\log(p)\log(nmq)}{\sqrt{n}m}  , \quad \delta \in \{1, 2\}.
\end{align*}
\end{theorem}

\begin{theorem}
\label{main.thm:Svc.e}
Under Assumption~\ref{main.as.A} and \ref{main.as.D}, we have $|\hat\sigma_e^2 - \sigma_e^{2, *}| = o_p(1)$.
\end{theorem}

Similar to the proof of Theorem~\ref{main.thm:1}, we leverage random matrix theory \parencite{tropp2015introduction} to bound the eigenvalues of matrix products and follow the classical proof of the consistency of the LASSO estimator \parencite{buhlmann2011statistics}. As previously mentioned, our estimator allows $q$ to be larger than $m$, and also applies to the case where $q$ is smaller than $m$. The related lemmas are collected in the Supplementary Materials. 

\printbibliography

\end{document}